\RequirePackage{fix-cm}
\documentclass[smallextended]{svjour3}       % onecolumn (second format)
\smartqed  % flush right qed marks, e.g. at end of proof
\usepackage{graphicx}
\usepackage{amsfonts}
\usepackage{amssymb}
\usepackage{amsfonts,mathrsfs,latexsym,amsmath,amssymb}
\usepackage[misc]{ifsym}
\begin{document}

\title{A new trace bilinear form on cyclic $\mathbb{F}_q$-linear $\mathbb{F}_{q^t}$-codes%\thanks{Grants or other notes
}
\subtitle{}


\author{Yun Gao \and Tingting Wu \and Fang-Wei Fu %etc.
}


\institute { \Letter ~Yun Gao\at
              Chern Institute of Mathematics and LPMC, Nankai University, Tianjin, 300071, China  \\
               \email{ gaoyun2014@126.com}
           \and Tingting Wu \at
            Chern Institute of Mathematics and LPMC, Nankai University, Tianjin, 300071, China  \\
               \email{ wutingting@mail.nankai.edu.com}
           \and Fang-Wei Fu \at
            Chern Institute of Mathematics and LPMC, Nankai University, Tianjin, 300071, China  \\
               \email{fwfu@nankai.edu.cn}
}

\date{Received: date / Accepted: date}

\maketitle

\begin{abstract}
Let $\mathbb{F}_q$ be a finite field of cardinality $q$, where $q$ is a power of a prime number $p$, $t\geq 2$ an even number satisfying $t \not\equiv 1 \;(\bmod \;p)$ and $\mathbb{F}_{q^t}$ an extension field of $\mathbb{F}_q$ with degree $t$. First, a new trace bilinear form on $\mathbb{F}_{{q^t}}^n$ which is called $\Delta$-bilinear form is given, where $n$ is a positive integer coprime to $q$. Then according to this new trace bilinear form, bases and enumeration of cyclic $\Delta$-self-orthogonal and cyclic $\Delta$-self-dual $\mathbb{F}_q$-linear $\mathbb{F}_{q^t}$-codes are investigated when $t=2$. Furthermore, some good $\mathbb{F}_q$-linear $\mathbb{F}_{q^2}$-codes are obtained.

\keywords{$\Delta$-bilinear form \and $\mathbb{F}_q$-linear $\mathbb{F}_{q^t}$-codes \and cyclic $\Delta$-self-orthogonal $\mathbb{F}_q$-linear $\mathbb{F}_{q^t}$-codes \and cyclic $\Delta$-self-dual $\mathbb{F}_q$-linear $\mathbb{F}_{q^t}$-codes}
\end{abstract}

\section{Introduction}
\label{intro}
Additive codes over $\mathbb{F}_4$ were first introduced in 1998 in \cite{Calderbank} connecting these codes to binary quantum codes. A year later, $\mathbb{F}_p$-linear codes over $\mathbb{F}_{p^2}$, where $p$ is prime, were connected in \cite{Rains} to nonbinary quantum codes. Additive codes were also generalized and studied in many papers, for example \cite{Bierbrauer,Cao1,Danielsen}.
\par
Let $\mathbb{F}_q$ be a finite field of cardinality $q$, where $q$ is a power of a prime number $p$, and $\mathbb{F}_{q^t}$ be an extension field of $\mathbb{F}_q$ with degree $t$.
An additive code over $\mathbb{F}_4$ is simply an $\mathbb{F}_2$-linear subspace of $\mathbb{F}_4^n$. A natural generalization of this is the following. An $\mathbb{F}_q$-linear $\mathbb{F}_{q^t}$-code $\mathcal {C}$ of length $n$ is an $\mathbb{F}_q$-linear subspace of $\mathbb{F}_{{q^t}}^n$. $\mathcal {C}$ is said to be cyclic if $({c_{n - 1}},{c_0},{c_1}, \cdots ,{c_{n - 2}}) \in \mathcal {C}$ for all $({c_0},{c_1}, \cdots ,{c_{n - 1}}) \in \mathcal {C}$. In particular, $\mathcal {C}$ is closed under componentwise addition and multiplication with elements from $\mathbb{F}_q$ \cite{Dey,Huffman2,Huffman3}.
\par
There have been extensive study and application of $\mathbb{F}_q$-linear $\mathbb{F}_{q^t}$-codes, where $t\geq 2$ is an integer.
Cao et al. \cite{Cao2} studied the structure and a canonical form
decomposition of any $\lambda$-constacyclic $\mathbb{F}_q$-linear code over $\mathbb{F}_{q^l}$. Furthermore, Cao et al. \cite{Cao3} gave the structure and canonical form decompositions of semisimple multivariable $\mathbb{F}_q$-linear codes over $\mathbb{F}_{q^l}$.
Huffman \cite{Huffman1} placed two different trace inner products on $\mathbb{F}_q$-linear $\mathbb{F}_{q^t}$-codes and examined the case $t = 2$ in detail giving specific bases and enumeration for the self-orthogonal and self-dual cyclic codes. This paper is a generalization of Huffman \cite{Huffman1}. $0$-trace bilinear form for all $t$ and $\gamma$-trace bilinear form for $t$ even on $\mathbb{F}_{{q^t}}^n$ have been studied in \cite{Huffman1}. Sharma and Kaur placed a new trace bilinear form on $\mathbb{F}_{{q^t}}^n$ which is called the $*$ bilinear form for all $t$, but they did not consider the $\Delta$-bilinear form for $t$ even.
In this paper, we give a new trace bilinear form on $\mathbb{F}_{{q^t}}^n$ which is called $\Delta$-bilinear form for $t$ even, where $n$ is a positive integer coprime to $q$, and study the bases and enumeration of cyclic $\Delta$-self-orthogonal and cyclic $\Delta$-self-dual $\mathbb{F}_q$-linear $\mathbb{F}_{q^2}$-codes. Our theory and method could be employed to obtain many good codes, which give the same parameters as the best known linear codes.
\par
The present paper is organized as follows. In Section 2, we sketch the basic Lemmas needed in this paper. Section 3 gives a new trace bilinear form on $\mathbb{F}_{{q^t}}^n$ which is called $\Delta$-bilinear form, where $n$ is a positive integer coprime to $q$. According to this new trace bilinear form, bases and enumeration of cyclic $\Delta$-self-orthogonal and cyclic $\Delta$-self-dual $\mathbb{F}_q$-linear $\mathbb{F}_{q^2}$-codes are investigated, respectively, in Section 4. Finally, we describe a program to construct cyclic $\Delta$-self-orthogonal and cyclic $\Delta$-self-dual $\mathbb{F}_3$-linear $\mathbb{F}_{3^2}$-codes of length $7$ and construct some good $\mathbb{F}_q$-linear $\mathbb{F}_{q^2}$-codes in Section 5.

\section{Preliminaries}
Let $\mathcal {R}_n^{(q)}$ and $\mathcal {R}_n^{(q^t)}$ denote the group algebra ${\mathbb{F}_q}[X]/\left\langle {{X^n} - 1} \right\rangle $ and ${\mathbb{F}_{q^t}}[X]/\left\langle {{X^n} - 1} \right\rangle $, respectively, where $X$ is an indeterminate over $\mathbb{F}_p$ and any extension field of $\mathbb{F}_p$, $n$ is a positive integer coprime to $q$ and $t\geq 2$ is an integer. As $\gcd (n,q) = 1$, $X^n-1$ has distinct roots and
$\mathcal {R}_n^{(q)}$, $\mathcal {R}_n^{(q^t)}$ are semi-simple. Furthermore, $\mathcal {R}_n^{(q)}$ and $\mathcal {R}_n^{(q^t)}$ can be written as a direct sum of its minimal ideals, respectively, all of which are fields. A minimal ideal is one which does not contain any smaller nonzero ideal.
If $\mathcal {C}$ is a linear code over $\mathbb{F}_{q^t}$ with parameters $[n,k,d]$, where $n$ is the length of $\mathcal {C}$, $k$ is the dimension of $\mathcal {C}$ and $d$ is the minimum Hamming distance of $\mathcal {C}$, then $\mathcal {C}$ is called maximum distance separable, or MDS for short, if $d = n - k + 1$.
Throughout this paper, ${\dim _K}V$ denotes the dimension of a finite-dimensional vector space $V$ over the field $K$.
\par
Let ${X^n} - 1 = {m_0}(X){m_1}(X) \cdots {m_{s - 1}}(X)$, where $m_i(X)$ is a monic irreducible polynomial over $\mathbb{F}_q$ for $0\leq i\leq s-1$ and $m_0(X)=-1+X$.
Let $\eta'$ be a fixed primitive $n$th root of unity over the splitting field of $X^n-1$ over
 $\mathbb{F}_p$. Then $C_{{l_i}}^{(q)} = \{ {l_i},{l_i}q,{l_i}{q^2}, \cdots \} \;(\bmod \;n)$ is the $q$-cyclotomic coset modulo $n$ and $l_i$ is chosen so that $\{ {{\eta'} ^k}|k \in C_{{l_i}}^{(q)}\} $ are the roots of $m_i(X)$ in a splitting field of $X^n-1$ over $\mathbb{F}_q$, where $0\leq i\leq s-1$.
\par
Next, we sum up some results from \cite{Huffman1} and restate it as Lemmas 1 and 2.
\par
\begin{lemma}{\cite[Lemma 1]{Huffman1}}
For any integer $l$, $C_l^{(q)} = C_l^{({q^t})} \cup C_{lq}^{({q^t})} \cup  \cdots  \cup C_{l{q^{a - 1}}}^{({q^t})}$, where the union is a disjoint union and $a = \gcd (t,| {C_l^{(q)}} |)$. Furthermore, for $0\leq i\leq a-1$, \[| {C_{l{q^i}}^{({q^t})}}| = | {C_l^{({q^t})}} | = \frac{{| {C_l^{(q)}} |}}{{\gcd (t,| {C_l^{(q)}} |)}}.\]
\end{lemma}
\par
\begin{lemma}{\cite[Theorem 1]{Huffman1}}
The following hold.\\
(i) ${X^n} - 1 = {m_0}(X){m_1}(X) \cdots {m_{s - 1}}(X)$, where $m_i(X)$ is a monic irreducible polynomial over $\mathbb{F}_q$ for $0\leq i\leq s-1$ with ${m_i}(X) \leftrightarrow C_{{l_i}}^{(q)}$ and ${m_0}(X) =  - 1 + X \leftrightarrow C_0^{(q)} = \{ 0\} $ and $\leftrightarrow$ gives the association between ${m_i}(X)$ and cyclotomic coset.\\
(ii) For $0\leq i\leq s-1$, ${m_i}(X) = {M_{i,0}}(X){M_{i,1}}(X) \cdots {M_{i,{s_i} - 1}}(X)$, where ${M_{i,j}}(X)$ is a monic irreducible polynomial over $\mathbb{F}_{q^t}$ with ${M_{i,j}}(X) \leftrightarrow C_{{l_i}{q^j}}^{({q^t})}$ for $0 \le j \le {s_i} - 1$. Also ${m_0}(X) =  - 1 + X = {M_{0,0}}(X) \leftrightarrow C_0^{({q^t})} = \{ 0\} $ with $s_0=1$. Furthermore, the factorization of $X^n-1$ into monic irreducible polynomials over $\mathbb{F}_{q^t}$ is given by ${X^n} - 1 = \prod\nolimits_{i = 0}^{s - 1} {\prod\nolimits_{j = 0}^{{s_i} - 1} {{M_{i,j}}(X)} } $.\\
(iii) For $0\leq i\leq s-1$, $\deg {m_i}(X) = | {C_{{l_i}}^{(q)}} |$. In addition, ${s_i} = \gcd (t,| {C_{{l_i}}^{(q)}} |)$ and $\deg {M_{i,j}}(X) = | {C_{{l_i}{q^j}}^{({q^t})}} | = {{| {C_{{l_i}}^{(q)}} |} {\left/
 {\vphantom {{| {C_{{l_i}}^{(q)}} |} {\gcd (t,| {C_{{l_i}}^{(q)}} |)}}} \right.
 \kern-\nulldelimiterspace} {\gcd (t,| {C_{{l_i}}^{(q)}} |)}}$ for $0 \le j \le {s_i} - 1$.\\
 (iv) $\mathcal {R}_n^{(q)} = {\mathcal {K}_0} \oplus {\mathcal {K}_1} \oplus  \cdots  \oplus {\mathcal {K}_{s - 1}}$, where ${\mathcal {K}_i} \leftrightarrow C_{{l_i}}^{(q)}$ and ${\mathcal {K}_i}$ is the ideal of $\mathcal {R}_n^{(q)}$ generated by ${\hat{m_i}}(X) = {{({X^n} - 1)} \mathord{\left/
 {\vphantom {{{X^n} - 1} {{m_i}(X)}}} \right.
 \kern-\nulldelimiterspace} {{m_i}(X)}}$. ${\mathcal {K}_i} \cong {\mathbb{F}_{{q^{{d_i}}}}}$, where ${d_i} = | {C_{{l_i}}^{(q)}} |$ and $\mathcal {K}_i\mathcal {K}_j=\{0\}$ if $i \neq j$.\\
 (v) $\mathcal {R}_n^{({q^t})} = {\mathcal {I}_{0,0}} \oplus {\mathcal {I}_{1,0}} \oplus  \cdots  \oplus {\mathcal {I}_{1,{s_1} - 1}} \oplus  \cdots  \oplus {\mathcal {I}_{s - 1,0}} \oplus  \cdots  \oplus {\mathcal {I}_{s - 1,{s_{s - 1}} - 1}}$, where ${\mathcal {I}_{i,j}}\leftrightarrow C_{{l_i}{q^j}}^{({q^t})}$ and ${\mathcal {I}_{i,j}}$ is the ideal of $\mathcal {R}_n^{(q^t)}$ generated by ${\hat{M}_{i,j}}(X) = {{({X^n} - 1)} \mathord{\left/
 {\vphantom {{({X^n} - 1)} {{M_{i,j}}(X)}}} \right.
 \kern-\nulldelimiterspace} {{M_{i,j}}(X)}}$. ${\mathcal {I}_{i,j}} \cong {\mathbb{F}_{{q^{t{D_i}}}}}$, where ${D_i} = | {C_{{l_i}{q^j}}^{({q^t})}} | = {{{d_i}} \mathord{\left/
 {\vphantom {{{d_i}} {{s_i}}}} \right.
 \kern-\nulldelimiterspace} {{s_i}}}$ for $0 \le i \le s - 1,\;0 \le j \le {s_i} - 1$ and ${\mathcal {I}_{i,j}}{\mathcal {I}_{i',j'}}=\{0\}$ if $(i,j)\neq (i',j')$.
\end{lemma}
\par
The relationship between the minimal ideals $\mathcal {K}_i$ and $\mathcal {I}_{i,j}$ determines the nature of the containment $\mathcal {R}_n^{(q)} \subset\mathcal {R}_n^{(q^t)}$. Huffman \cite{Huffman1} defined the ring automorphism ${\tau _{{q^w},u}}:\;\mathcal {R}_n^{({q^r})} \to \mathcal {R}_n^{({q^r})}$ as ${\tau _{{q^w},u}}(\sum\limits_{k = 0}^{n - 1} {{a_k}{X^k}} ) = \sum\limits_{k = 0}^{n - 1} {a_k^{{q^w}}{X^{uk}}} ,$ where $w, r, u$ are integers with $0 \le w \le r,\;1 \le u \le n - 1$ and $\gcd (u,n) = 1$. When $r=t$, the map ${\tau _{{q^w},u}}:\;\mathcal {R}_n^{({q^t})} \to \mathcal {R}_n^{({q^t})}$ is a ring automorphism and permutes the minimal ideals $\mathcal {I}_{i,j}$ given in Lemma 2 (v). This permutation action can be described completely by the cyclotomic cosets.
\par
\begin{lemma}{\cite[Lemma 2]{Huffman1}}
For $0 \le i \le s - 1$ and $0 \le j \le {s_i} - 1$, we have ${\tau _{{q^w},u}}({\mathcal {I}_{i,j}}) = {\mathcal {I}_{i',j'}}$, where $i'$ and $j'$ are determined as follows. If ${\mathcal {I}_{i,j}} \leftrightarrow C_l^{({q^t})}$, then ${\mathcal {I}_{i',j'}} \leftrightarrow C_{l{u^{ - 1}}{q^w}}^{({q^t})}$.
\end{lemma}
\par
In \cite{Huffman1}, we have ${\mathcal {K}_i} \subset {\mathcal {J}_i} = {\mathcal {I}_{i,0}} \oplus {\mathcal {I}_{i,1}} \oplus  \cdots  \oplus {\mathcal {I}_{i,{s_i} - 1}}$ for $0 \le i \le s - 1$. So we can determine the precise containment.
\par
\begin{lemma}{\cite[Theorem 2]{Huffman1}}
For $0 \le i \le s - 1$, \[\begin{array}{l}
{\mathcal {K}_i} = \{ f(X) \in {\mathcal {J}_i}|{\tau _{q,1}}(f(X)) = f(X)\} \\
\;\;\;\;\; = \{ c(X) + {\tau _{q,1}}(c(X)) +  \cdots  + {\tau _{{q^{{s_i} - 1}},1}}(c(X))|c(X) \in {\mathcal {I}_{i,0}}\;\\
\;\;\;\;\;\;\;\;\;{\rm{and}}\;{\tau _{{q^{{s_i}}},1}}(c(X)) = c(X)\} .
\end{array}\]
\end{lemma}
\par
Now, we define the trace map $T{r_{Q,r}}:\;{\mathbb{F}_{{Q^r}}} \to {\mathbb{F}_Q}$ by $T{r_{Q,r}}(b) = \sum\nolimits_{w = 0}^{r - 1} {{b^{{Q^w}}}} $ for $b\in \mathbb{F}_{{Q^r}}$, where $Q$ is a power of $q$ and $r$ is a positive integer.
Let $t=2^am$, $A=2^{a-1}$ and $Q=q^A$, where $a\geq 1$ and $m$ is odd. Then $t=2Am$. Since $2A\left| t \right.$, ${\mathbb{F}_{{Q^2}}} = {\mathbb{F}_{{q^{2A}}}}$ is a subfield of $\mathbb{F}_{q^t}$. By \cite{Huffman1}, there exists an element $0\neq\gamma  \in {\mathbb{F}_{{Q^2}}} \subseteq {\mathbb{F}_{{q^t}}}$ such that $T{r_{Q,2}}(\gamma ) = \gamma  + {\gamma ^Q} = 0$, so $\gamma^Q=-\gamma$. When $t=2$, we have $\gamma^q=-\gamma$. Huffman \cite{Huffman1} determined the Hermitian trace inner product as follows.
\[{\left\langle {a,b} \right\rangle _\gamma } = \sum\limits_{i = 0}^{n - 1} {T{r_{q,t}}(\gamma {a_i}b_i^{{q^{{t \mathord{\left/
 {\vphantom {t 2}} \right.
 \kern-\nulldelimiterspace} 2}}}}) = } \sum\limits_{i = 0}^{n - 1} {\sum\limits_{w = 0}^{t - 1} {{{(\gamma {a_i}b_i^{{q^{{t \mathord{\left/
 {\vphantom {t 2}} \right.
 \kern-\nulldelimiterspace} 2}}}})}^{{q^w}}}} } ,\]
 where $a = ({a_0},{a_1}, \cdots ,{a_{n - 1}}),\;b = ({b_0},{b_1}, \cdots, {b_{n - 1}} )\in \mathbb{F}_{{q^t}}^n$.
 Then, the Hermitian trace bilinear form on $\mathcal {R}_n^{(q^t)}$ was determined as bellow.
 \[{(a(X),b(X))_\gamma } = \sum\limits_{w = 0}^{t - 1} {{\tau _{{q^w},1}}} (\gamma a(X){\tau _{{q^{{t \mathord{\left/
 {\vphantom {t 2}} \right.
 \kern-\nulldelimiterspace} 2}}}, - 1}}(b(X))),\]
 where $a(X)=a_0+a_1X+ \cdots +a_{n-1}X^{n-1},\;b(X)=b_0+b_1X+ \cdots +b_{n-1}X^{n-1} \in \mathcal {R}_n^{({q^t})}$.
 \par
If $a=(a_0,a_1,\ldots, a_{n-1})\in \mathbb{F}_{{q^t}}^n$, define $\sigma (a) = (a_{n-1},a_0,a_1,\ldots , a_{n-2})$ is the cyclic shift of $a$. A permutation $\mu$ of $\{0,1,2,\ldots,s-1\}$ is defined by ${\tau _{1, - 1}}({\mathcal {J}_i}) = {\tau _{{q^{{t \mathord{\left/
 {\vphantom {t 2}} \right.
 \kern-\nulldelimiterspace} 2}}}, - 1}}({\mathcal {J}_i}) = {\mathcal {J}_{\mu (i)}}$ for $0 \le i \le s - 1$.  At the end of this section, we list several basic Lemmas 5-7 needed in the following sections.
 \par
\begin{lemma}{\cite[Theorem 7]{Huffman1}}
Let $\mathcal {C} = {\mathcal {C}_0} \oplus {\mathcal {C}_1} \oplus  \cdots  \oplus {\mathcal {C}_{s - 1}}$ and ${\mathcal {C}^{{ \bot _\gamma }}} = \mathcal {C}_0^{'} \oplus \mathcal {C}_1^{'} \oplus  \cdots  \oplus \mathcal {C}_{s - 1}^{'}$, where ${\mathcal {C}^{{ \bot _\gamma }}} = \{ v \in \mathbb{F}_{{q^t}}^n|{\langle c,v\rangle_\gamma } = 0\;{\rm{for}}\;{\rm{all}}\;c \in \mathcal {C}\} $, ${\mathcal {C}_i} = \mathcal {C} \cap {\mathcal {J}_i}$ and $\mathcal {C}_i^{'} = {\mathcal {C}^{{ \bot _\gamma }}} \cap {\mathcal {J}_i}$ for all $0\leq i\leq s-1$. Then, for each $i$, we have \[\mathcal {C}_{\mu (i)}^{'} = \{ a(X) \in {\mathcal {J}_{\mu (i)}}|{(c(X),a(X))_\gamma } = 0\;{\rm{for}}\;{\rm{all}}\;c(X) \in {\mathcal {C}_i}\} .\]
Furthermore, if the $\mathcal {K}_i$-dimension of $\mathcal {C}_i$ is $k_i$, then the $\mathcal {K}_{\mu(i)}$-dimension of $\mathcal {C}_{\mu(i)}^{'}$ is $t-k_i$.
 \end{lemma}
\par
\begin{lemma}{\cite[Lemma 10]{Huffman1}}
Suppose that $\mu(i)=i$, where $1 \le i \le s - 1$. The following hold.\\
(i) $d_i$ is even unless $n$ is even and $l_i={n \mathord{\left/
 {\vphantom {n 2}} \right.
 \kern-\nulldelimiterspace} 2}$.\\
(ii) If $n$ is even, for some $1 \le i^{\#} \le s - 1$, then ${l_{{i^\# }}} = {n \mathord{\left/
 {\vphantom {n 2}} \right.
 \kern-\nulldelimiterspace} 2}$, $C_{{l_{{i^\# }}}}^{(q)} = C_{{l_{{i^\# }}}}^{({q^t})} = \{ {n \mathord{\left/
 {\vphantom {n 2}} \right.
 \kern-\nulldelimiterspace} 2}\} $, ${d_{{i^\# }}} = {s_{{i^\# }}} = 1$ and $\mu ({i^\# }) = {i^\# }$.\\
 (iii) If $i\neq i^{\#}$ and $s_i$ is odd, then ${\tau _{1, - 1}}({\mathcal {I}_{i,j}}) = {\mathcal {I}_{i,j}}$ for all $0\leq j<s_i$.\\
 (iv) If $i\neq i^{\#}$ and $s_i$ is even, then either ${\tau _{1, - 1}}({\mathcal {I}_{i,j}}) = {\mathcal {I}_{i,j}}$ for all $0\leq j<s_i$ or ${\tau _{1, - 1}}({\mathcal {I}_{i,j}}) = {\mathcal {I}_{i,j+{{{s_i}} \mathord{\left/
 {\vphantom {{{s_i}} 2}} \right.
 \kern-\nulldelimiterspace} 2}}}$ for all $0\leq j<s_i$, where the second subscript is computed modulo $s_i$.\\
 (v) If $i\neq i^{\#}$ and ${\tau _{1, - 1}}({\mathcal {I}_{i,0}}) = {\mathcal {I}_{i,0}}$, then $D_i$ is even and $ - {l_i} \equiv {l_i}{q^{{{{D_i}} \mathord{\left/
 {\vphantom {{{D_i}} 2}} \right.
 \kern-\nulldelimiterspace} 2}}}\;(\bmod \;n)$.
 \end{lemma}
\par
\begin{lemma}{\cite[Lemma 11]{Huffman1}}
If $n$ is even, let $i^{\#}$ be chosen so that ${l_{{i^\# }}} = {n \mathord{\left/
 {\vphantom {n 2}} \right.
 \kern-\nulldelimiterspace} 2}$. The following hold.\\
 (i) ${\tau _{1, - 1}}$ is the identity map on $\mathcal {J}_0$ and ${\tau _{q,1}}(a(X)) = a{(X)^q}$ for all $a(X)\in \mathcal {J}_0$.\\
 (ii) If $n$ is even, then ${\tau _{1, - 1}}$ is the identity map on ${\mathcal {J}_{{i^\# }}} = {\mathcal {I}_{{i^\# },0}}$ and ${\tau _{q,1}}(a(X)) = a{(X)^q}$ for all $a(X)\in \mathcal {J}_{i^\#}$.\\
 (iii) Suppose $i\neq 0$ and $i\neq i^\#$. If ${\tau _{1, - 1}}({\mathcal {I}_{i,j}}) = {\mathcal {I}_{i,j}}$, then $\mu(i)=i$ and ${\tau _{1, - 1}}(a(X)) = a{(X)^{{q^{{{t{D_i}} \mathord{\left/
 {\vphantom {{t{D_i}} 2}} \right.
 \kern-\nulldelimiterspace} 2}}}}}$ for all $a(X)\in \mathcal {I}_{i,j}$.\\
 (iv) Suppose $i\neq 0$, $i\neq i^\#$ and $t$ is even. If ${\tau _{{q^{{t \mathord{\left/
 {\vphantom {t 2}} \right.
 \kern-\nulldelimiterspace} 2}}}, - 1}}({\mathcal {I}_{i,j}}) = {\mathcal {I}_{i,j}}$, then $\mu(i)=i$ and ${\tau _{{q^{{t \mathord{\left/
 {\vphantom {t 2}} \right.
 \kern-\nulldelimiterspace} 2}}}, - 1}}(a(X)) = a{(X)^{{q^{{{t{D_i}} \mathord{\left/
 {\vphantom {{t{D_i}} 2}} \right.
 \kern-\nulldelimiterspace} 2}}}}}$ for all $a(X)\in \mathcal {I}_{i,j}$.\\
 (v) If $i\neq 0$, $i\neq i^\#$ and $t$ is even, there does not exist $j$ with $0\leq j< s_i$ such that both ${\tau _{1, - 1}}({\mathcal {I}_{i,j}}) = {\mathcal {I}_{i,j}}$ and ${\tau _{{q^{{t \mathord{\left/
 {\vphantom {t 2}} \right.
 \kern-\nulldelimiterspace} 2}}}, - 1}}({\mathcal {I}_{i,j}}) = {\mathcal {I}_{i,j}}$ hold.
 \end{lemma}

\section{$\Delta$-Trace bilinear form on $\mathbb{F}_{{q^t}}^n$ and $\mathcal {R}_n^{(q^t)}$} \label{}
\noindent
 In this section, we define a new trace bilinear form on $\mathbb{F}_{{q^t}}^n$ and $\mathcal {R}_n^{(q^t)}$ for any even integer $t \geq 2$ satisfying $t \not\equiv 1 \;(\bmod \;p)$ and study its properties. We need the following Lemma.
\par
\begin{lemma}
  Let $t=2^am\geq 2$ be an even integer satisfying $t \not\equiv 1\; (\bmod \;p)$, where $a \geq 1$ and $m$ is odd. Then there exists an element $\gamma  \in {\mathbb{F}_{{q^{{2^a}}}}} \subseteq {\mathbb{F}_{{q^t}}}$ with $\gamma  \ne 0$ such that $\gamma  + {\gamma ^{{q^{{2^{a - 1}}}}}} = 0$. Thus the map $\psi :\;{\mathbb{F}_{{q^t}}} \to {\mathbb{F}_{{q^t}}}$ defined as $\psi (\alpha ) = {\alpha ^q} + {\alpha ^{{q^2}}} +  \cdots + {\alpha ^{{q^{t - 1}}}} = T{r_{q,t}}(\alpha ) - \alpha $ is an $\mathbb{F}_q$-linear vector space automorphism, where $\alpha  \in {\mathbb{F}_{{q^t}}}$.
\end{lemma}
\begin{proof}
The result will be proved if we can show $\psi$ is an injective $\mathbb{F}_q$-linear vector space automorphism. It can easily be verified that $\psi$ is an $\mathbb{F}_q$-linear map. It remains to show that the kernel of $\psi$ denoted by ${\rm{Ker(}}\psi {\rm{)}}$ is $\{0\}$. If $\alpha  \in {\rm{Ker(}}\psi {\rm{)}}$, then $\psi (\alpha ) = {\alpha ^q} + {\alpha ^{{q^2}}} +  \cdots  + {\alpha ^{{q^{t - 1}}}} = 0$. Furthermore, we have $\psi {(\alpha )^q} = \alpha  + {\alpha ^{{q^2}}} +  \cdots  + {\alpha ^{{q^{t - 1}}}} = 0$. Then, we obtain $\psi (\alpha ) - \psi {(\alpha )^q} = {\alpha ^q} - \alpha  = 0$, which implies ${\alpha ^{{q^u}}} = \alpha $ for each integer $u\;( 1\leq u\leq t-1)$. This implies that $(t - 1)\alpha  = \psi (\alpha ) = 0$. As $t \not\equiv 1\; (\bmod \;p)$, we have $(t - 1)\alpha  = 0$ if and only if $\alpha=0$. The proof is completed.
\end{proof}

Now, we define a trace inner product ${( \cdot , \cdot )_\Delta }:\;\mathbb{F}_{{q^t}}^n \times \mathbb{F}_{{q^t}}^n \to {\mathbb{F}_q}$ as $${(a,b)_\Delta } = \sum\limits_{j = 0}^{n - 1} {T{r_{q,t}}} (\gamma {a_j}\psi (b_j^{{q^{{t \mathord{\left/
 {\vphantom {t 2}} \right.
 \kern-\nulldelimiterspace} 2}}}}))$$ for all $a = ({a_0},{a_1}, \cdots ,{a_{n - 1}})$, $b = ({b_0},{b_1}, \cdots ,{b_{n - 1}})$ in $\mathbb{F}_{{q^t}}^n$. In the following Lemma, we prove that the trace inner product $( \cdot , \cdot )_\Delta$ is a non-degenerate trace bilinear form on $\mathbb{F}_{{q^t}}^n$ for any prime power $q$.
\par
\begin{lemma}
  For $a,b,c \in \mathbb{F}_{{q^t}}^n$ and $\alpha \in \mathbb{F}_q$, the following hold.\\
   (i) ${(a,b)_\Delta } \in \mathbb{F}_q$.\\
   (ii) ${(a,b + c)_\Delta } = {(a,b)_\Delta } + {(a,c)_\Delta }$ and ${(a + b,c)_\Delta } = {(a,c)_\Delta } + {(b,c)_\Delta }$.\\
   (iii) ${(\alpha a,b)_\Delta } = {(a,\alpha b)_\Delta } = \alpha {(a,b)_\Delta }$.\\
   (iv) ${( \cdot , \cdot )_\Delta }$ is non-degenerate.\\
\end{lemma}
\begin{proof}
To prove these results, we let $a = ({a_0},{a_1}, \cdots ,{a_{n - 1}}), b = ({b_0},{b_1}, \cdots ,{b_{n - 1}})$ and $c = ({c_0},{c_1}, \cdots ,{c_{n - 1}})$, where ${a_i},{b_i},{c_i} \in {\mathbb{F}_{{q^t}}}$ for all $0 \leq i \leq n-1$.\\
\indent
(i) As ${(a,b)_\Delta } = \sum\limits_{i= 0}^{n - 1} {T{r_{q,t}}} (\gamma {a_i}\psi (b_i^{{q^{{t \mathord{\left/
 {\vphantom {t 2}} \right.
 \kern-\nulldelimiterspace} 2}}}}))$ and ${T{r_{q,t}}} (\gamma {a_i}\psi (b_i^{{q^{{t \mathord{\left/
 {\vphantom {t 2}} \right.
 \kern-\nulldelimiterspace} 2}}}})) \in \mathbb{F}_q$ for $0 \leq i \leq n-1$, (i) holds.\\
\indent
Parts (ii) and (iii) follow because $T{r_{q,t}}$ is $\mathbb{F}_q$-linear.\\
\indent
(iv) To prove this, we need to show that if $( a , b )_\Delta =0 $ for all $b \in \mathbb{F}_{{q^t}}^n$, then $a=0$. Suppose, on the contrary, that there exists a non-zero $a = ({a_0},{a_1}, \cdots ,{a_{n - 1}}) \in \mathbb{F}_{{q^t}}^n$ such that $( a , b )_\Delta =0 $ for all $b \in \mathbb{F}_{{q^t}}^n$. In particular, let ${a_j} \ne 0$ for some $j$. As $T{r_{q,t}}$ is an onto map and $\gamma \neq 0$, there exists $d \in {\mathbb{F}_{{q^t}}}$  such that ${T{r_{q,t}}} (\gamma {a_j}\psi (d^{{q^{{t \mathord{\left/
 {\vphantom {t 2}} \right.
 \kern-\nulldelimiterspace} 2}}}}))\neq 0$. Let $b = ({b_0},{b_1}, \cdots ,{b_{n - 1}})\in \mathbb{F}_{{q^t}}^n$, where $b_i=0$ for all $i\neq j$ and $b_j=d$. Then,  ${T{r_{q,t}}} (\gamma {a_j}\psi (b_j^{{q^{{t \mathord{\left/
 {\vphantom {t 2}} \right.
 \kern-\nulldelimiterspace} 2}}}}))\neq 0$. This leads to a contradiction.\\
\end{proof}

Next, we define a form ${[ \cdot , \cdot ]_\Delta}:\;\mathcal {R}_n^{({q^t})} \times \mathcal {R}_n^{({q^t})} \rightarrow \mathcal {R}_n^{(q)}$ as follows. \[{[a(X),b(X)]_\Delta } = \sum\limits_{u = 0}^{t - 1} {{\tau _{{q^u},1}}} (\gamma a(X)\sum\limits_{w = 1}^{t - 1} {{\tau _{{q^{{t \mathord{\left/
 {\vphantom {t 2}} \right.
 \kern-\nulldelimiterspace} 2} + w}}, - 1}}} (b(X)))\] for all $a(X),b(X) \in \mathcal {R}_n^{({q^t})}$. We prove a result for this form analogous to Lemma 9 as follows.
\par
\begin{lemma}
   Let $a(X),b(X),c(X) \in \mathcal {R}_n^{({q^t})}$. The following hold.\\
   (i) ${[a(X),b(X)]_\Delta } = \sum\limits_{k = 0}^{n - 1} {{{(a,{\sigma ^k}(b))}_\Delta }} {X^k}$.\\
   (ii) ${[a(X),b(X)]_\Delta } \in \mathcal {R}_n^{(q)}$.\\
   (iii) ${[a(X),b(X) + c(X)]_\Delta } = {[a(X),b(X)]_\Delta } + {[a(X),c(X)]_\Delta }$ and $${[a(X) + b(X),c(X)]_\Delta } = {[a(X),c(X)]_\Delta } + {[b(X),c(X)]_\Delta }.$$\\
   (iv) For $f(X) \in \mathcal {R}_n^{(q)}$, we have ${[f(X)a(X),b(X)]_\Delta } = f(X){[a(X),b(X)]_\Delta }$ and $${[a(X),f(X)b(X)]_\Delta } = {\tau _{1, - 1}}(f(X)){[a(X),b(X)]_\Delta }.$$\\
   (v) ${[ \cdot , \cdot ]_\Delta }$ is non-degenerate.
\end{lemma}

\begin{proof}
(i) According to the definition of ${[ \cdot , \cdot ]_\Delta }$, it follows that ${[a(X),b(X)]_\Delta } = \sum\limits_{u = 0}^{t - 1} {{\tau _{{q^u},1}}} (\gamma a(X)\sum\limits_{w = 1}^{t - 1} {{\tau _{{q^{{t \mathord{\left/
 {\vphantom {t 2}} \right.
 \kern-\nulldelimiterspace} 2} + w}}, - 1}}} (b(X)))$. For $\gamma a(X)\sum\limits_{w = 1}^{t - 1} {{\tau _{{q^{{t \mathord{\left/
 {\vphantom {t 2}} \right.
 \kern-\nulldelimiterspace} 2} + w}}, - 1}}} (b(X))$, we have the following results.\\
 \indent
(1) When $w=1$, we have \[\begin{array}{l}
\gamma a(X){\tau _{{q^{{t \mathord{\left/
 {\vphantom {t 2}} \right.
 \kern-\nulldelimiterspace} 2} + 1}}, - 1}}(b(X))\\
\;\; = \gamma ({a_0} + {a_1}X + {a_2}{X^2} + {a_3}{X^3} +  \cdots  + {a_{n - 2}}{X^{n - 2}} + {a_{n - 1}}{X^{n - 1}})\\
\;\;\;\;\;\; \cdot ({b_0}^{{q^{{t \mathord{\left/
 {\vphantom {t 2}} \right.
 \kern-\nulldelimiterspace} 2} + 1}}} + b_{n - 1}^{{q^{{t \mathord{\left/
 {\vphantom {t 2}} \right.
 \kern-\nulldelimiterspace} 2} + 1}}}X + b_{n - 2}^{{q^{{t \mathord{\left/
 {\vphantom {t 2}} \right.
 \kern-\nulldelimiterspace} 2} + 1}}}{X^2} +  \cdots  + b_2^{{q^{{t \mathord{\left/
 {\vphantom {t 2}} \right.
 \kern-\nulldelimiterspace} 2} + 1}}}{X^{n - 2}} + b_1^{{q^{{t \mathord{\left/
 {\vphantom {t 2}} \right.
 \kern-\nulldelimiterspace} 2} + 1}}}{X^{n - 1}})\\
\; \;= {c_{0,1}} + {c_{1,1}}X + {c_{2,1}}{X^2} +  \cdots  + {c_{n - 2,1}}{X^{n - 2}} + {c_{n - 1,1}}{X^{n - 1}},
\end{array}\]
where
\[{c_{0,1}} = \gamma ({a_0}b_0^{{q^{{t \mathord{\left/
 {\vphantom {t 2}} \right.
 \kern-\nulldelimiterspace} 2} + 1}}} + {a_1}b_1^{{q^{{t \mathord{\left/
 {\vphantom {t 2}} \right.
 \kern-\nulldelimiterspace} 2} + 1}}} + {a_2}b_2^{{q^{{t \mathord{\left/
 {\vphantom {t 2}} \right.
 \kern-\nulldelimiterspace} 2} + 1}}} +  \cdots  + {a_{n - 2}}b_{n - 2}^{{q^{{t \mathord{\left/
 {\vphantom {t 2}} \right.
 \kern-\nulldelimiterspace} 2} + 1}}} + {a_{n - 1}}b_{n - 1}^{{q^{{t \mathord{\left/
 {\vphantom {t 2}} \right.
 \kern-\nulldelimiterspace} 2} + 1}}});\]
 \[{c_{1,1}} = \gamma ({a_0}b_{n - 1}^{{q^{{t \mathord{\left/
 {\vphantom {t 2}} \right.
 \kern-\nulldelimiterspace} 2} + 1}}} + {a_1}b_0^{{q^{{t \mathord{\left/
 {\vphantom {t 2}} \right.
 \kern-\nulldelimiterspace} 2} + 1}}} + {a_2}b_1^{{q^{{t \mathord{\left/
 {\vphantom {t 2}} \right.
 \kern-\nulldelimiterspace} 2} + 1}}} +  \cdots  + {a_{n - 2}}b_{n - 3}^{{q^{{t \mathord{\left/
 {\vphantom {t 2}} \right.
 \kern-\nulldelimiterspace} 2} + 1}}} + {a_{n - 1}}b_{n - 2}^{{q^{{t \mathord{\left/
 {\vphantom {t 2}} \right.
 \kern-\nulldelimiterspace} 2} + 1}}});\]
 \[{c_{2,1}} = \gamma ({a_0}b_{n - 2}^{{q^{{t \mathord{\left/
 {\vphantom {t 2}} \right.
 \kern-\nulldelimiterspace} 2} + 1}}} + {a_1}b_{n - 1}^{{q^{{t \mathord{\left/
 {\vphantom {t 2}} \right.
 \kern-\nulldelimiterspace} 2} + 1}}} + {a_2}b_0^{{q^{{t \mathord{\left/
 {\vphantom {t 2}} \right.
 \kern-\nulldelimiterspace} 2} + 1}}} +  \cdots  + {a_{n - 2}}b_{n - 4}^{{q^{{t \mathord{\left/
 {\vphantom {t 2}} \right.
 \kern-\nulldelimiterspace} 2} + 1}}} + {a_{n - 1}}b_{n - 3}^{{q^{{t \mathord{\left/
 {\vphantom {t 2}} \right.
 \kern-\nulldelimiterspace} 2} + 1}}});\]
\begin{center}
$ \cdots $
\end{center}
\[{c_{n - 2,1}} = \gamma ({a_0}b_2^{{q^{{t \mathord{\left/
 {\vphantom {t 2}} \right.
 \kern-\nulldelimiterspace} 2} + 1}}} + {a_1}b_3^{{q^{{t \mathord{\left/
 {\vphantom {t 2}} \right.
 \kern-\nulldelimiterspace} 2} + 1}}} + {a_2}b_4^{{q^{{t \mathord{\left/
 {\vphantom {t 2}} \right.
 \kern-\nulldelimiterspace} 2} + 1}}} +  \cdots  + {a_{n - 2}}b_0^{{q^{{t \mathord{\left/
 {\vphantom {t 2}} \right.
 \kern-\nulldelimiterspace} 2} + 1}}} + {a_{n - 1}}b_1^{{q^{{t \mathord{\left/
 {\vphantom {t 2}} \right.
 \kern-\nulldelimiterspace} 2} + 1}}});\]
 \[{c_{n - 1,1}} = \gamma ({a_0}b_1^{{q^{{t \mathord{\left/
 {\vphantom {t 2}} \right.
 \kern-\nulldelimiterspace} 2} + 1}}} + {a_1}b_2^{{q^{{t \mathord{\left/
 {\vphantom {t 2}} \right.
 \kern-\nulldelimiterspace} 2} + 1}}} + {a_2}b_3^{{q^{{t \mathord{\left/
 {\vphantom {t 2}} \right.
 \kern-\nulldelimiterspace} 2} + 1}}} +  \cdots  + {a_{n - 2}}b_{n - 1}^{{q^{{t \mathord{\left/
 {\vphantom {t 2}} \right.
 \kern-\nulldelimiterspace} 2} + 1}}} + {a_{n - 1}}b_0^{{q^{{t \mathord{\left/
 {\vphantom {t 2}} \right.
 \kern-\nulldelimiterspace} 2} + 1}}}).\]
 \indent
 (2) When $w=2$, we have
 \[\begin{array}{l}
\gamma a(X){\tau _{{q^{{t \mathord{\left/
 {\vphantom {t 2}} \right.
 \kern-\nulldelimiterspace} 2} + 2}}, - 1}}(b(X))\\
\;\;  = \gamma ({a_0} + {a_1}X + {a_2}{X^2} + {a_3}{X^3} +  \cdots  + {a_{n - 2}}{X^{n - 2}} + {a_{n - 1}}{X^{n - 1}})\\
\;\;\;\;\;\; \cdot ({b_0}^{{q^{{t \mathord{\left/
 {\vphantom {t 2}} \right.
 \kern-\nulldelimiterspace} 2} + 2}}} + b_{n - 1}^{{q^{{t \mathord{\left/
 {\vphantom {t 2}} \right.
 \kern-\nulldelimiterspace} 2} + 2}}}X + b_{n - 2}^{{q^{{t \mathord{\left/
 {\vphantom {t 2}} \right.
 \kern-\nulldelimiterspace} 2} + 2}}}{X^2} +  \cdots  + b_2^{{q^{{t \mathord{\left/
 {\vphantom {t 2}} \right.
 \kern-\nulldelimiterspace} 2} + 2}}}{X^{n - 2}} + b_1^{{q^{{t \mathord{\left/
 {\vphantom {t 2}} \right.
 \kern-\nulldelimiterspace} 2} + 2}}}{X^{n - 1}})\\
\;\;  = {c_{0,2}} + {c_{1,2}}X + {c_{2,2}}{X^2} +  \cdots  + {c_{n - 2,2}}{X^{n - 2}} + {c_{n - 1,2}}{X^{n - 1}},
\end{array}\]
where
\[{c_{0,2}} = \gamma ({a_0}b_0^{{q^{{t \mathord{\left/
 {\vphantom {t 2}} \right.
 \kern-\nulldelimiterspace} 2} + 2}}} + {a_1}b_1^{{q^{{t \mathord{\left/
 {\vphantom {t 2}} \right.
 \kern-\nulldelimiterspace} 2} + 2}}} + {a_2}b_2^{{q^{{t \mathord{\left/
 {\vphantom {t 2}} \right.
 \kern-\nulldelimiterspace} 2} + 2}}} +  \cdots  + {a_{n - 2}}b_{n - 2}^{{q^{{t \mathord{\left/
 {\vphantom {t 2}} \right.
 \kern-\nulldelimiterspace} 2} + 2}}} + {a_{n - 1}}b_{n - 1}^{{q^{{t \mathord{\left/
 {\vphantom {t 2}} \right.
 \kern-\nulldelimiterspace} 2} + 2}}});\]
\[{c_{1,2}} = \gamma ({a_0}b_{n - 1}^{{q^{{t \mathord{\left/
 {\vphantom {t 2}} \right.
 \kern-\nulldelimiterspace} 2} + 2}}} + {a_1}b_0^{{q^{{t \mathord{\left/
 {\vphantom {t 2}} \right.
 \kern-\nulldelimiterspace} 2} + 2}}} + {a_2}b_1^{{q^{{t \mathord{\left/
 {\vphantom {t 2}} \right.
 \kern-\nulldelimiterspace} 2} + 2}}} +  \cdots  + {a_{n - 2}}b_{n - 3}^{{q^{{t \mathord{\left/
 {\vphantom {t 2}} \right.
 \kern-\nulldelimiterspace} 2} + 2}}} + {a_{n - 1}}b_{n - 2}^{{q^{{t \mathord{\left/
 {\vphantom {t 2}} \right.
 \kern-\nulldelimiterspace} 2} + 2}}});\]
\[{c_{2,2}} = \gamma ({a_0}b_{n - 2}^{{q^{{t \mathord{\left/
 {\vphantom {t 2}} \right.
 \kern-\nulldelimiterspace} 2} + 2}}} + {a_1}b_{n - 1}^{{q^{{t \mathord{\left/
 {\vphantom {t 2}} \right.
 \kern-\nulldelimiterspace} 2} + 2}}} + {a_2}b_0^{{q^{{t \mathord{\left/
 {\vphantom {t 2}} \right.
 \kern-\nulldelimiterspace} 2} + 2}}} +  \cdots  + {a_{n - 2}}b_{n - 4}^{{q^{{t \mathord{\left/
 {\vphantom {t 2}} \right.
 \kern-\nulldelimiterspace} 2} + 2}}} + {a_{n - 1}}b_{n - 3}^{{q^{{t \mathord{\left/
 {\vphantom {t 2}} \right.
 \kern-\nulldelimiterspace} 2} + 2}}});\]
 \begin{center}
$ \cdots $
\end{center}
\[{c_{n - 2,2}} = \gamma ({a_0}b_2^{{q^{{t \mathord{\left/
 {\vphantom {t 2}} \right.
 \kern-\nulldelimiterspace} 2} + 2}}} + {a_1}b_3^{{q^{{t \mathord{\left/
 {\vphantom {t 2}} \right.
 \kern-\nulldelimiterspace} 2} + 2}}} + {a_2}b_4^{{q^{{t \mathord{\left/
 {\vphantom {t 2}} \right.
 \kern-\nulldelimiterspace} 2} + 2}}} +  \cdots  + {a_{n - 2}}b_0^{{q^{{t \mathord{\left/
 {\vphantom {t 2}} \right.
 \kern-\nulldelimiterspace} 2} + 2}}} + {a_{n - 1}}b_1^{{q^{{t \mathord{\left/
 {\vphantom {t 2}} \right.
 \kern-\nulldelimiterspace} 2} + 2}}});\]
\[{c_{n - 1,2}} = \gamma ({a_0}b_1^{{q^{{t \mathord{\left/
 {\vphantom {t 2}} \right.
 \kern-\nulldelimiterspace} 2} + 2}}} + {a_1}b_2^{{q^{{t \mathord{\left/
 {\vphantom {t 2}} \right.
 \kern-\nulldelimiterspace} 2} + 2}}} + {a_2}b_3^{{q^{{t \mathord{\left/
 {\vphantom {t 2}} \right.
 \kern-\nulldelimiterspace} 2} + 2}}} +  \cdots  + {a_{n - 2}}b_{n - 1}^{{q^{{t \mathord{\left/
 {\vphantom {t 2}} \right.
 \kern-\nulldelimiterspace} 2} + 2}}} + {a_{n - 1}}b_0^{{q^{{t \mathord{\left/
 {\vphantom {t 2}} \right.
 \kern-\nulldelimiterspace} 2} + 2}}}).\]
 The rest, $3 \le w \le t - 1$, can be done in the same manner. Therefore, we have
 \[\gamma a(X)\sum\limits_{w = 1}^{t - 1} {{\tau _{{q^{{t \mathord{\left/
 {\vphantom {t 2}} \right.
 \kern-\nulldelimiterspace} 2} + w}}, - 1}}} (b(X)) = \sum\limits_{k = 0}^{n - 1} {\left( {\sum\limits_{\scriptstyle\;\;\;\;\;\;\;\;i,j = 0\hfill\atop
\scriptstyle i - j \equiv k\;(\bmod n)\hfill}^{n - 1} {\gamma {a_i}\psi (b_j^{{q^{{t \mathord{\left/
 {\vphantom {t 2}} \right.
 \kern-\nulldelimiterspace} 2}}}}){X^k}} } \right)} .\]
 Then, we have
\[\begin{array}{l}
{[a(X),b(X)]_\Delta } = \sum\limits_{u = 0}^{t - 1} {{\tau _{{q^u},1}}} \left( {\sum\limits_{k = 0}^{n - 1} {\left( {\sum\limits_{\scriptstyle\;\;\;\;\;\;\;\;i,j = 0\hfill\atop
\scriptstyle i - j \equiv k\;(\bmod n)\hfill}^{n - 1} {\gamma {a_i}\psi (b_j^{{q^{{t \mathord{\left/
 {\vphantom {t 2}} \right.
 \kern-\nulldelimiterspace} 2}}}}){X^k}} } \right)} } \right) \\
\;\;\;\;\;\;\;\;\;\;\;\;\;\;\;\;\;\;\;\;\;\;\; = {(a,b)_\Delta } + {(a,\sigma (b))_\Delta }X +  \cdots  + {(a,{\sigma ^{(n - 1)}}(b))_\Delta }{X^{n - 1}},
\end{array}\]
 where the cyclic shift $\sigma (b) = ({b_{n - 1}},{b_0},{b_1}, \cdots ,{b_{n - 2}})$ of $b= ({b_0},{b_1}, \cdots ,{b_{n - 1}}) \in \mathbb{F}_{{q^t}}^n$ is identified with $Xb(X) \in \mathcal {R}_n^{({q^t})}$, (i) holds.\\
\indent
 (ii) It follows from part (i) and Lemma 9 (i).\\
\indent
 (iii) As ${\tau _{{q^u},1}}$ and ${\tau _{{q^u},-1}}$ are ring automorphisms for any integer $u \geq 0$, part (iii) follows immediately.\\
\indent
 (iv) For $f(X) \in \mathcal {R}_n^{(q)}$, as ${\tau _{q,1}}(f(X)) = f(X)$, we have ${\tau _{{q^{{t \mathord{\left/
 {\vphantom {t 2}} \right.
 \kern-\nulldelimiterspace} 2}}}, - 1}}(f(X)) = {\tau _{1, - 1}}(f(X))$. From this, (iv) holds.\\
\indent
 (v) Part (v) will be proved by showing that if ${[a(X),b(X)]_\Delta } = 0$ for all $b(X) \in \mathcal {R}_n^{({q^t})}$, then $a(X)=0$. If not, there would exists $j\; (0\leq j \leq n-1)$ such that $a_j\neq 0$. Since $T{r_{q,t}}$ is an onto map and $\gamma \neq 0$, there exists $\theta \in {\mathbb{F}_{{q^t}}}$  such that $T{r_{q,t}}(\gamma\theta)\neq 0$. Then for $b(X) = {\psi ^{ - 1}}(\theta a_j^{ - 1}{X^{-j}})$, let $a = ({a_0},{a_1}, \cdots ,{a_{n - 1}})$ where $a_i=0$ for all $i \neq j$ and $a_j \neq 0$. By part (i), we have ${[a(X),b(X)]_\Delta } = T{r_{q,t}}(\gamma\theta)\neq 0$, which is a contradiction.
\end{proof}

Now we proceed to study the dual codes of cyclic $\mathbb{F}_q$-linear $\mathbb{F}_{q^t}$-code $\mathcal {C}$ of length $n$ with respect to this new trace bilinear form ${( \cdot , \cdot )_\Delta }$ on $\mathbb{F}_{{q^t}}^n$.

%%%%%%%%%%%%%%%%%%%%%%%%%%%%%%%%%%%%%%%%%%%%%%%%%%%%%%%%%%%%%%%%%%%%%%%%%%%%%%%%%%%%%%%%%%
\section{Dual codes of cyclic $\mathbb{F}_q$-linear $\mathbb{F}_{q^t}$-codes} \label{}
\noindent
Let $\mathcal {C}$ be a cyclic $\mathbb{F}_q$-linear $\mathbb{F}_{q^t}$-code of length $n$, where $\gcd (n,q) = 1$.
Then, the $\Delta$-dual code of $\mathcal {C}$ is defined as ${\mathcal {C}^{{ \bot _\Delta }}} = \{ v \in \mathbb{F}_{{q^t}}^n|{(c,v)_\Delta } = 0\;{\rm{for}}\;{\rm{all}}\;c \in \mathcal {C}\} $. It is easy to verify that the dual code $\mathcal {C}^{{ \bot _\Delta }}$ is also an $\mathbb{F}_q$-linear $\mathbb{F}_{q^t}$-code of length $n$. Furthermore, if $\mathcal {C}$ is cyclic, then its dual code $\mathcal {C}^{{ \bot _\Delta }}$ is also cyclic. From now onwards, throughout this paper, we will view the cyclic $\mathbb{F}_q$-linear $\mathbb{F}_{q^t}$-code $\mathcal {C}$ of length $n$ and its $\Delta$-dual code ${\mathcal {C}^{{ \bot _\Delta }}}$ as $\mathcal {R}_n^{(q)}$-submodule of $\mathcal {R}_n^{(q^t)}$, respectively. In addition, if $\mathcal {C} \subseteq\mathcal {R}_n^{(q^t)}$ is any cyclic $\mathbb{F}_q$-linear $\mathbb{F}_{q^t}$-code, then one can easily obtain its dual code  ${\mathcal {C}^{{ \bot _\Delta }}}\subseteq\mathcal {R}_n^{(q^t)}$ with respect to trace bilinear form ${[.,.]_\Delta }$.
\par
Next, we study the properties of the $\Delta$-dual codes of cyclic $\mathbb{F}_q$-linear $\mathbb{F}_{q^t}$-codes of length $n$. For the ring automorphism ${\tau _{{q^u}, - 1}}\;(0 \le u \le t - 1)$ on the ideal $\mathcal {J}_i\;(0 \le i \le s - 1)$ of $\mathcal {R}_n^{(q^t)}$, we observe that $\mathcal {C}_{ - {l_0}}^{(q)} = \mathcal {C}_{{l_0}}^{(q)}$, and further for each $i \;(1 \le i \le s - 1)$, there exists a unique integer $i' \;(1 \le i' \le s - 1)$ satisfying $\mathcal {C}_{ - {l_i}}^{(q)} = \mathcal {C}_{{l_{i'}}}^{(q)}$. This gives rise to a permutation $\mu$ of $\{0,1,2,\ldots,s-1\}$ defined by $\mathcal {C}_{ - {l_i}}^{(q)} = \mathcal {C}_{{l_{\mu(i)}}}^{(q)}$ for $0 \le i \le s - 1$. It is not difficult to show that $\mu(0)=0$ and $\mu(\mu(i))=i$ for $0 \le i \le s - 1$. That is, $\mu$ is either the identity permutation or a product of transpositions. When $n$ is even, by Lemma 6, there exists an integer ${i^\# }\;(1 \le {i^\# } \le s - 1)$ satisfying $\mathcal {C}_{ {l_{i^\# }}}^{(q)} = \mathcal {C}_{\frac{n}{2}}^{(q)}=\{\frac{n}{2}\}$, as $q$ is odd. Note that $\mathcal {C}_{ -{l_{i^\# }}}^{(q)} = \mathcal {C}_{l_{i^\# }}^{(q)}$, so $\mu ({i^\# }) = {i^\# }$. With the help of the above concepts, we give the following Lemma.
\par
\begin{lemma}
Let $u$ $(0\leq u\leq t-1)$ be a fixed integer. Then we have ${\tau _{{q^u}, - 1}}({\mathcal {J}_i}) = {\mathcal {J}_{\mu(i)}}$ for $0\leq i\leq s-1$.
\end{lemma}
\begin{proof}
Working in a similar way as in Lemma 8 of Huffman \cite{Huffman1}, the result follows.
\end{proof}
\par
We are interested in the relationship between $\mathcal {C}$ and ${\mathcal {C}^{{ \bot _\Delta }}}$ viewed as $\mathcal {R}_n^{(q)}$-submodules of $\mathcal {R}_n^{(q^t)}$. The following result will prove quite useful.
\par
\begin{lemma}
Let ${a_i}(X),\;{b_i}(X) \in {\mathcal {J}_i}$ for $0\leq i\leq s-1$. If $a(X) = \sum\nolimits_{i = 0}^{s - 1} {{a_i}(X)} $ and $b(X) = \sum\nolimits_{i = 0}^{s - 1} {{b_i}(X)}$, then ${[a(X),b(X)]_\Delta } = \sum\nolimits_{i = 0}^{s - 1} {[{a_i}(X)} ,{b_{\mu (i)}}(X){]_\Delta }.$
\end{lemma}

\begin{proof}
By Lemma 10 (iii), we have ${[a(X),b(X)]_\Delta } = \sum\nolimits_{i = 0}^{s - 1} {\sum\nolimits_{j = 0}^{s - 1} {{{[{a_i}(X),{b_j}(X)]}_\Delta }} } $. It suffices to prove the result in the case that ${[{a_i}(X),{b_j}(X)]_\Delta } = 0$ if $j \ne \mu (i)$. By Lemma 11, we have $\sum\limits_{w = 1}^{t - 1} {{\tau _{{q^{{t \mathord{\left/
 {\vphantom {t 2}} \right.
 \kern-\nulldelimiterspace} 2} + w}}, - 1}}} (b_j(X)) \in {\mathcal {J}_{\mu (j)}}$. If $j \ne \mu (i)$, then we have $\mu (j) \ne \mu (\mu (i)) = i$. By Lemma 2, we have that $\mathcal {J}_i{\mathcal {J}_{\mu (j)}}=\{0\}$ implying $\gamma a_i(X)\sum\limits_{w = 1}^{t - 1} {{\tau _{{q^{{t \mathord{\left/
 {\vphantom {t 2}} \right.
 \kern-\nulldelimiterspace} 2} + w}}, - 1}}} (b_j(X))=0$. By the definition of ${[.,.]_\Delta }$, this shows that ${[{a_i}(X),{b_j}(X)]_\Delta } = 0$ if $j \ne \mu (i)$.
\end{proof}
\par
\begin{corollary}
Let $a,b \in \mathbb{F}_{{q^t}}^n$ be identified with $a(X) = \sum\nolimits_{i = 0}^{s - 1} {{a_i}(X)} $ and $b(X) = \sum\nolimits_{i = 0}^{s - 1} {{b_i}(X)} $, respectively, where $a_i(X), b_i(X)\in \mathcal {J}_i$ for $0\leq i\leq s-1$. Then $b\in {\mathcal {C}^{{ \bot _\Delta }}}$ if and only if ${[{a_i}(X),{b_{\mu (i)}}(X)]_\Delta } = 0$ for all  $0\leq i\leq s-1$ and all $a\in \mathcal {C}.$
\end{corollary}

\begin{proof}
Since $\mathcal {C}$ is cyclic, $b\in {\mathcal {C}^{{ \bot _\Delta }}}$ if and only if ${(a,b)_\Delta } = 0$ for all $a\in \mathcal {C}$ if and only if ${(a,{\sigma ^k}(b))_\Delta } = 0$ for all $a\in \mathcal {C}$ and any integer $k$ if and only if ${[a(X),b(X)]_\Delta } = 0$ for all $a(X)\in \mathcal {C}$ by Lemma 10 (i). By Lemma 12, we have ${[a(X),b(X)]_\Delta } = \sum\nolimits_{i = 0}^{s - 1} {[{a_i}(X)} ,{b_{\mu (i)}}(X){]_\Delta }$. For $0\leq i\leq s-1$, by definitions of $\mu$ and ${[.,.]_\Delta }$, we have ${[{a_i}(X),{b_{\mu (i)}}(X)]_\Delta }\in \mathcal {J}_i$. Therefore, as $\mathcal {R}_n^{({q^t})} = {\mathcal {J}_0} \oplus {\mathcal {J}_1} \oplus  \cdots  \oplus {\mathcal {J}_{s - 1}}$ is a direct sum, ${[a(X),b(X)]_\Delta } = \sum\nolimits_{i = 0}^{s - 1} {[{a_i}(X)} ,{b_{\mu (i)}}(X){]_\Delta }=0$ if and only if ${[{a_i}(X),{b_{\mu (i)}}(X)]_\Delta } = 0$ for all $0\leq i\leq s-1$. The proof is completed.
\end{proof}
\par
According to Lemma 5, we have the following Theorem.
\par
\begin{theorem}
Let $\mathcal {C}$ be a cyclic $\mathbb{F}_q$-linear $\mathbb{F}_{q^t}$-code of length $n$ and ${\mathcal {C}^{{ \bot _\Delta }}}$ the dual code of $\mathcal {C}$. Then we have $\mathcal {C} = {\mathcal {C}_0} \oplus {\mathcal {C}_1} \oplus  \cdots  \oplus {\mathcal {C}_{s - 1}}$ and ${\mathcal {C}^{{ \bot _\Delta }}} = \mathcal {C}_0^{(\Delta )} \oplus \mathcal {C}_1^{(\Delta )} \oplus  \cdots  \oplus \mathcal {C}_{s - 1}^{(\Delta )}$, where ${\mathcal {C}_i} = \mathcal {C} \cap {\mathcal {J}_i}$ and $\mathcal {C}_i^{(\Delta )} = {\mathcal {C}^{{ \bot _\Delta }}} \cap {\mathcal {J}_i}$ for all $0\leq i\leq s-1$. Thus, for each $i$, we have \[\mathcal {C}_{\mu (i)}^{(\Delta )} = \{ a(X) \in {\mathcal {J}_{\mu (i)}}|{[c(X),a(X)]_\Delta } = 0\;{for}\;{all}\;c(X) \in {\mathcal {C}_i}\} .\]
Furthermore, if the $\mathcal {K}_i$-dimension of $\mathcal {C}_i$ is $k_i$, then the $\mathcal {K}_{\mu(i)}$-dimension of $\mathcal {C}_{\mu(i)}^{(\Delta )}$ is $t-k_i$.
\end{theorem}
\begin{proof}
The proof of this result is quite similar to Theorem 7 of \cite{Huffman1} and so is omitted.
\end{proof}
\par
In addition, a cyclic $\mathbb{F}_q$-linear $\mathbb{F}_{q^t}$-code $\mathcal {C}$ of length $n$ is said to be cyclic $\Delta$-self-orthogonal if it satisfies $\mathcal {C} \subseteq {\mathcal {C}^{{ \bot _\Delta }}}$, and it is called cyclic $\Delta$-self-dual if it satisfies $\mathcal {C} = {\mathcal {C}^{{ \bot _\Delta }}}$.
Cyclic $\Delta$-self-orthogonal and cyclic $\Delta$-self-dual $\mathbb{F}_q$-linear $\mathbb{F}_{q^t}$-codes are characterized in the following Lemma.
\par
\begin{lemma}
Let $\mathcal {C}$ be a cyclic $\mathbb{F}_q$-linear $\mathbb{F}_{q^t}$-code of length $n$ and ${\mathcal {C}^{{ \bot _\Delta }}}$ the dual code of $\mathcal {C}$. Let us write $\mathcal {C} = {\mathcal {C}_0} \oplus {\mathcal {C}_1} \oplus  \cdots  \oplus {\mathcal {C}_{s - 1}}$ and ${\mathcal {C}^{{ \bot _\Delta }}} = \mathcal {C}_0^{(\Delta )} \oplus \mathcal {C}_1^{(\Delta )} \oplus  \cdots  \oplus \mathcal {C}_{s - 1}^{(\Delta )}$, where ${\mathcal {C}_i} = \mathcal {C} \cap {\mathcal {J}_i}$ and $\mathcal {C}_i^{(\Delta )} = {\mathcal {C}^{{ \bot _\Delta }}} \cap {\mathcal {J}_i}$ for all $0\leq i\leq s-1$. Then,\\
(i) $\mathcal {C}$ is cyclic $\Delta$-self-orthogonal if and only if $\mathcal {C}_i \subseteq \mathcal {C}_i^{(\Delta )}$ for all $0\leq i\leq s-1$.\\
(ii) $\mathcal {C}$ is cyclic $\Delta$-self-dual if and only if $\mathcal {C}_i =\mathcal {C}_i^{(\Delta )}$ for all $0\leq i\leq s-1$.
\end{lemma}
\begin{proof}
The proof is trivial.
\end{proof}
\par
Now we consider the case $t=2$ and study cyclic $\Delta$-self-orthogonal and cyclic $\Delta$-self-dual $\mathbb{F}_q$-linear $\mathbb{F}_{q^2}$-codes. It is not difficult to find that, when $t=2$, the minimal ideal ${\mathcal {I}_{i,j}}$ of $\mathcal {R}_n^{({q^2})}$ is the finite field of order ${q^{2{D_i}}}$ for $0\leq i\leq s-1$ and $0\leq j\leq s_i-1$, where ${D_i} = \frac{{{d_i}}}{{{s_i}}}$ and ${s_i} = \gcd (2,{d_i})$. For $0\leq i\leq s-1$, we choose primitive elements ${\rho _{i,0}}(X),{\rho _{i,1}}(X), \cdots ,{\rho _{i,{s_i} - 1}}(X)$ of ${\mathcal {I}_{i,0}},{\mathcal {I}_{i,1}}, \cdots ,{\mathcal {I}_{i,{s_i} - 1}}$, respectively, satisfying
${\tau _{{q^j},1}}({\rho _{i,0}}(X)) = {\rho _{i,j}}(X)$ and $e_{i,j}(X)$ the identity of ${\mathcal {I}_{i,j}}$ for all $0\leq j\leq s_i-1$.

%%%%%%%%%%%%%%%%%%%%%%%%%%%%%%%%%%%%%%%%%%%%%%%%%%%%%%%%%%%%%%%%%%%%%%%%%%%%%%%%%%%%%%%%%%
\subsection{Cyclic $\Delta$-self-orthogonal $\mathbb{F}_q$-linear $\mathbb{F}_{q^2}$-codes} \label{}
\noindent
The bases of all cyclic $\Delta$-self-orthogonal $\mathbb{F}_q$-linear $\mathbb{F}_{q^2}$-codes of length $n$ are determined in the following Theorem, where $\gcd (n,q) = 1$.
\par
\begin{theorem}
Let $t=2$, $q$ be a power of the prime $p$ and $\gcd (n,q) = 1$. Let $\mathcal {C}$ be a cyclic $\mathbb{F}_q$-linear $\mathbb{F}_{q^2}$-code of length $n$ and ${\mathcal {C}^{{ \bot _\Delta }}}$ the dual code of $\mathcal {C}$. We write $\mathcal {C} = {\mathcal {C}_0} \oplus {\mathcal {C}_1} \oplus  \cdots  \oplus {\mathcal {C}_{s - 1}}$ and ${\mathcal {C}^{{ \bot _\Delta }}} = \mathcal {C}_0^{(\Delta )} \oplus \mathcal {C}_1^{(\Delta )} \oplus  \cdots  \oplus \mathcal {C}_{s - 1}^{(\Delta )}$, where ${\mathcal {C}_i} = \mathcal {C} \cap {\mathcal {J}_i}$ and $\mathcal {C}_i^{(\Delta )} = {\mathcal {C}^{{ \bot _\Delta }}} \cap {\mathcal {J}_i}$ for all $0\leq i\leq s-1$. Then $\mathcal {C}$ is cyclic $\Delta$-self-orthogonal if and only if for each $i$ $(0\leq i\leq s-1)$, the following hold.\\
(i) When $i=0$ or, if $n$ is even, $i=i^\#$, then \\
\indent
(a) $\mathcal {C}_i=\{0\}$, or\\
\indent
(b) $\mathcal {C}_i$ is a 1-dimensional $\mathcal {K}_i$-subspace of $\mathcal {J}_i$ with basis $\{ {\rho _{i,0}}{(X)^k}\} $, where $k=0$ when $q$ is even and $k=\frac{{q + 1}}{2}$ when $q$ is odd.\\
(ii) When $i\neq 0$, $i\neq i^\#$, $\mu(i)=i$ and ${\tau _{1, - 1}}({\mathcal {I}_{i,0}}) = {\mathcal {I}_{i,0}}$, then\\
\indent
(a) $\mathcal {C}_i=\{0\}$, or\\
\indent
(b) $\mathcal {C}_i$ is 1-dimensional over $\mathcal {K}_i$ with basis $\{ {e_{i,0}}(X) + {\rho _{i,1}}{(X)^k}\} $, where $k = ({q^{{{{d_i}} \mathord{\left/
 {\vphantom {{{d_i}} 2}} \right.
 \kern-\nulldelimiterspace} 2}}} - 1)m$ for $0\leq m\leq {q^{{{{d_i}} \mathord{\left/
 {\vphantom {{{d_i}} 2}} \right.
 \kern-\nulldelimiterspace} 2}}}$.\\
 (iii) When $i\neq 0$, $i\neq i^\#$, $\mu(i)=i$ and ${\tau _{1, - 1}}({\mathcal {I}_{i,0}}) = {\mathcal {I}_{i,1}}$, then\\
\indent
 (a) $\mathcal {C}_i=\{0\}$, or\\
\indent
(b) $\mathcal {C}_i$ is 1-dimensional over $\mathcal {K}_i$ with bases $\{ {e_{i,0}}(X) \}$, $\{ {e_{i,1}}(X) \}$ and $\{ {e_{i,0}}(X) + {\rho _{i,1}}{(X)^k}\} $, where $k = ({q^{{{{d_i}} \mathord{\left/
 {\vphantom {{{d_i}} 2}} \right.
 \kern-\nulldelimiterspace} 2}}} + 1)m$ for $0 \le m \le {q^{{{{d_i}} \mathord{\left/
 {\vphantom {{{d_i}} 2}} \right.
 \kern-\nulldelimiterspace} 2}}} - 2$.\\
 (iv) When $\mu(i)\neq i$ and $d_i$ is odd,\\
 \indent
 (a) if $\mathcal {C}_i=\{0\}$, then ${\mathcal {C}_{\mu (i)}} \subseteq {\mathcal {J}_{\mu (i)}}$;\\
\indent
 (b) if $\mathcal {C}_i=\mathcal {J}_i$, then ${\mathcal {C}_{\mu (i)}}=\{0\} $;\\
\indent
 (c)  if $\mathcal {C}_i$ is 1-dimensional over $\mathcal {K}_i$ with basis $\{{\rho _{i,0}}{(X)^k}\} $, where $0 \le k \le {q^{{d_i}}}$, then ${\mathcal {C}_{\mu (i)}} \subseteq \mathcal {C}_{\mu (i)}^{(\Delta )}$ with $\{ {\rho _{\mu (i),0}}{(X)^{k'}}\} $ the basis of $\mathcal {C}_{\mu (i)}^{(\Delta )}$ over $\mathcal {K}_{\mu(i)}$, where $k'=k=0$ or $k'=q^{d_i}+1-k$ for $1 \le k \le {q^{{d_i}}}$.\\
 (v)  When $\mu(i)\neq i$ and $d_i$ is even,\\
 \indent
 (a) if $\mathcal {C}_i=\{0\}$, then ${\mathcal {C}_{\mu (i)}} \subseteq {\mathcal {J}_{\mu (i)}}$;\\
\indent
 (b) if $\mathcal {C}_i=\mathcal {J}_i$, then ${\mathcal {C}_{\mu (i)}}=\{0\} $;\\
\indent
 (c)  if $\mathcal {C}_i$ is 1-dimensional over $\mathcal {K}_i$ with basis $\{ a(X)\} $, then ${\mathcal {C}_{\mu (i)}} \subseteq \mathcal {C}_{\mu (i)}^{(\Delta )}$ with $\{ a'(X)\} $ the basis of $\mathcal {C}_{\mu (i)}^{(\Delta )}$ over $\mathcal {K}_{\mu(i)}$, where $a'(X) = {e_{\mu (i),0}}(X) + {\rho _{\mu (i),1}}{(X)^{k'}}$ when $a(X) = {e_{i,0}}(X) + {\rho _{i,1}}{(X)^k}$ with
  $k'=k=0$ or $k'=q^{d_i}-1-k$ for $1 \le k \le {q^{{d_i}}-2}$.
\end{theorem}
\begin{proof}
(i) Suppose that $\mu(i)=i$. Since ${\mathcal {C}_i} \subseteq \mathcal {C}_i^{(\Delta )} = \mathcal {C}_{\mu (i)}^{(\Delta )}$ and the sum of the $\mathcal {K}_i$-dimensions of $\mathcal {C}_i$ and $\mathcal {C}_{\mu (i)}^{(\Delta )}$ is at most 2, by Lemma 5, either $\mathcal {C}_i$ is $\{0\}$ or $\mathcal {C}_i$ is 1-dimensional over $\mathcal {K}_i$ and ${\mathcal {C}_i} = \mathcal {C}_i^{(\Delta )}$.\\
\indent
Assume that $i=0$ or, if $n$ is even, $i=i^\#$. In this case, by Lemma 2 (iv), (v) and Lemma 6 (ii), we have $d_i=s_i=D_i=1$, ${\mathcal {J}_i} = {\mathcal {I}_{i,0}} \cong {\mathbb{F}_{{q^2}}}$ and $\mathcal {K}_i\cong \mathbb{F}_q$. By the proof of Theorem 8 (i) of Huffman \cite{Huffman1}, we have $\{ {\rho _{i,0}}{(X)^k}\} $ are bases of the $q+1$ distinct 1-dimensional $\mathcal {K}_i$-subspaces of $\mathcal {J}_i$, where $0\leq k\leq q$. It is simple to show that $\{ {\rho _{i,0}}{(X)^k}\} $ are bases of $\mathcal {C}_i$, where $0\leq k\leq q$. As ${\mathcal {C}_i} = \mathcal {C}_i^{(\Delta )}$, we have ${[{\rho _{i,0}}{(X)^k},{\rho _{i,0}}{(X)^k}]_\Delta } = 0$. By Lemma 7 (i) and (ii), we have ${\tau _{1, - 1}}(c(X)) = c(X)$
and ${\tau _{q, 1}}(c(X)) = c(X)^q$ for $c(X)\in \mathcal {J}_i$. Then, we have ${\tau _{q, -1}}(c(X)) = c(X)^q$ for $c(X)\in \mathcal {J}_i$. We will use the observation that ${\tau _{q,1}}(\gamma ) = {\gamma ^q} =  - \gamma $ throughout the proof.\\
\indent
It is straightforward to show that (a) holds. For (b), we have
$$\begin{array}{l}
{[a(X),b(X)]_\Delta } = {\tau _{1,1}}(\gamma a(X){\tau _{q^2, - 1}}(b(X))) + {\tau _{q,1}}(\gamma a(X){\tau _{q^2, - 1}}(b(X)))\\
\;\;\;\;\;\;\;\;\;\;\;\;\;\;\;\;\;\;\;\;\;\;\; = \gamma a(X){\tau _{q^2, - 1}}(b(X)) - \gamma {\tau _{q,1}}(a(X){\tau _{q^2, - 1}}(b(X))).
\end{array}$$
So for $0\leq k\leq q$, \[{[{\rho _{i,0}}{(X)^k},{\rho _{i,0}}{(X)^k}]_\Delta } = \gamma {\rho _{i,0}}{(X)^k}{\rho _{i,0}}{(X)^{kq^2}} - \gamma {\rho _{i,0}}{(X)^{kq}}{\rho _{i,0}}{(X)^{k{q^3}}}.\]
As ${\rho _{i,0}}{(X)^{{q^2}}} = {\rho _{i,0}}(X) \in {\mathcal {I}_{i,0}} \cong {\mathbb{F}_{{q^2}}}$, we have $${[{\rho _{i,0}}{(X)^k},{\rho _{i,0}}{(X)^k}]_\Delta } = \gamma {\rho _{i,0}}{(X)^k}{\rho _{i,0}}{(X)^{k}} - \gamma {\rho _{i,0}}{(X)^{kq}}{\rho _{i,0}}{(X)^{kq}} = 0$$ if and only if ${\rho _{i,0}}{(X)^{2k(q-1)}}=e_{i,0}(X)$ if and only if $2k(q-1) \equiv 0\;(\bmod \;{q^2} - 1)$ if and only if $2k \equiv 0\;(\bmod\; q +1)$. When $q$ is even, we have $k \equiv 0\;(\bmod\; q + 1)$. Thus, we have $k=0$. When $q$ is odd, we have $k \equiv 0\;(\bmod \;\frac{{q + 1}}{2})$. Since $0\leq k\leq q$, we have $k=\frac{{q + 1}}{2}$.\\
\indent
(ii) Now consider the case when $\mu(i)=i$ but $i \notin \{ 0,{i^\# }\} $. By Lemma 6 (i) and Lemma 2, we have $d_i$ is even implying ${s_i} = \gcd (2,{d_i}) = 2$ and $t{D_i} = 2{D_i} = {d_i}$ as ${D_i} = {{{d_i}} \mathord{\left/
 {\vphantom {{{d_i}} {{s_i}}}} \right.
 \kern-\nulldelimiterspace} {{s_i}}}$. Hence, by Lemma 3, ${\mathcal {J}_i} = {\mathcal {I}_{i,0}} \oplus {\mathcal {I}_{i,1}}$ and ${\tau _{q,1}}({\mathcal {I}_{i,0}}) = {\mathcal {I}_{i,1}}$. Using Lemma 2 (iv) and (v), we have ${\mathcal {I}_{i,0}} \cong {\mathcal {I}_{i,1}} \cong {\mathcal {K}_i} \cong {\mathbb{F}_{{q^{{d_i}}}}}$. As in part (i), we need the bases of the different 1-dimensional $\mathcal {K}_i$-subspaces of $\mathcal {J}_i$. Since ${\tau _{{q^2},1}}$ is the identity on $\mathcal {R}_n^{({q^2})}$, by Lemma 4, we have ${\mathcal {K}_i} = \{ c(X) + {\tau _{q,1}}(c(X))|c(X) \in {\mathcal {I}_{i,0}}\}  \cong {\mathbb{F}_{{q^{{d_i}}}}}$. Therefore, the bases of the $q^{d_i}+1$ 1-dimensional $\mathcal {K}_i$-subspaces of $\mathcal {J}_i$ are $\{ {e_{i,0}}(X)\} $, $\{ {e_{i,1}}(X)\} $ and $\{ {e_{i,0}}(X) + {\rho _{i,1}}{(X)^k}\} $, where $0\leq k \leq q^{d_i}-2$.\\
 \indent
 According to part (i), if $\mathcal {C}_i$ is 1-dimensional over $\mathcal {K}_i$, then ${\mathcal {C}_i} = \mathcal {C}_i^{(\Delta )}$. Assume that ${\tau _{1, - 1}}({\mathcal {I}_{i,j}}) = {\mathcal {I}_{i,j}}$ for $0\leq j \leq 1$. By Lemma 7 (iii), we have ${\tau _{1, - 1}}(c(X)) = c{(X)^{{q^{{{t{D_i}} \mathord{\left/
 {\vphantom {{t{D_i}} 2}} \right.
 \kern-\nulldelimiterspace} 2}}}}} = c{(X)^{{q^{{{{d_i}} \mathord{\left/
 {\vphantom {{{d_i}} 2}} \right.
 \kern-\nulldelimiterspace} 2}}}}}$ for $c(X)$ in either ${\mathcal {I}_{i,0 }}$ or ${\mathcal {I}_{i,1}}$. Recall that
 $$\begin{array}{l}
{[a(X),b(X)]_\Delta } = \gamma a(X){\tau _{q^2, - 1}}(b(X)) - \gamma {\tau _{q,1}}(a(X){\tau _{q^2, - 1}}(b(X))).
\end{array}$$
Firstly, since \[\begin{array}{l}
{[{e_{i,0}}(X),{e_{i,0}}(X)]_\Delta } = \gamma {e_{i,0}}(X){\tau _{{q^2}, - 1}}({e_{i,0}}(X)) - \gamma {\tau _{q,1}}({e_{i,0}}(X){\tau _{{q^2}, - 1}}({e_{i,0}}(X)))\\
\;\;\;\;\;\;\;\;\;\;\;\;\;\;\;\;\;\;\;\;\;\;\;\;\;\;\;\;\;\; = \gamma {e_{i,0}}(X){e_{i,0}}{(X)^{{q^{{{{d_i}} \mathord{\left/
 {\vphantom {{{d_i}} 2}} \right.
 \kern-\nulldelimiterspace} 2}}}}} - \gamma {e_{i,1}}(X){e_{i,1}}{(X)^{{q^{{{{d_i}} \mathord{\left/
 {\vphantom {{{d_i}} 2}} \right.
 \kern-\nulldelimiterspace} 2}}}}}\\
\;\;\;\;\;\;\;\;\;\;\;\;\;\;\;\;\;\;\;\;\;\;\;\;\;\;\;\;\;\; = \gamma {e_{i,0}}(X) - \gamma {e_{i,1}}(X) \ne 0.
\end{array}\]
Similarly, we have ${[{e_{i,1}}(X),{e_{i,1}}(X)]_\Delta } \ne 0$. So neither  $\{ {e_{i,0}}(X) \}$ nor $\{ {e_{i,1}}(X) \}$ is a basis of $\mathcal {C}_i$. Secondly, we have \[\begin{array}{l}
{[{e_{i,0}}(X) + {\rho _{i,1}}{(X)^k},{e_{i,0}}(X) + {\rho _{i,1}}{(X)^k}]_\Delta }\\
\;\; = \gamma ({e_{i,0}}(X) + {\rho _{i,1}}{(X)^k})({e_{i,0}}{(X)^{{q^{{{{d_i}} \mathord{\left/
 {\vphantom {{{d_i}} 2}} \right.
 \kern-\nulldelimiterspace} 2}}}}} + {\rho _{i,1}}{(X)^{k{q^{{{{d_i}} \mathord{\left/
 {\vphantom {{{d_i}} 2}} \right.
 \kern-\nulldelimiterspace} 2}}}}})\\
\;\;\;\;\;\; - \gamma ({e_{i,1}}(X) + {\rho _{i,0}}{(X)^k})({e_{i,1}}{(X)^{{q^{{{{d_i}} \mathord{\left/
 {\vphantom {{{d_i}} 2}} \right.
 \kern-\nulldelimiterspace} 2}}}}} + {\rho _{i,0}}{(X)^{k{q^{{{{d_i}} \mathord{\left/
 {\vphantom {{{d_i}} 2}} \right.
 \kern-\nulldelimiterspace} 2}}}}})\\
\;\; = \gamma ({e_{i,0}}(X) + {\rho _{i,1}}{(X)^{k({q^{{{{d_i}} \mathord{\left/
 {\vphantom {{{d_i}} 2}} \right.
 \kern-\nulldelimiterspace} 2}}} + 1)}}) - \gamma ({e_{i,1}}(X) + {\rho _{i,0}}{(X)^{k({q^{{{{d_i}} \mathord{\left/
 {\vphantom {{{d_i}} 2}} \right.
 \kern-\nulldelimiterspace} 2}}} + 1)}})\\
\;\; = \gamma ({e_{i,0}}(X) - {\rho _{i,0}}{(X)^{k({q^{{{{d_i}} \mathord{\left/
 {\vphantom {{{d_i}} 2}} \right.
 \kern-\nulldelimiterspace} 2}}} + 1)}}) - \gamma ({e_{i,1}}(X) - {\rho _{i,1}}{(X)^{k({q^{{{{d_i}} \mathord{\left/
 {\vphantom {{{d_i}} 2}} \right.
 \kern-\nulldelimiterspace} 2}}} + 1)}}).
\end{array}\]
Then, we have ${[{e_{i,0}}(X) + {\rho _{i,1}}{(X)^k},{e_{i,0}}(X) + {\rho _{i,1}}{(X)^k}]_\Delta } = 0$ if and only if ${\rho _{i,0}}{(X)^{k({q^{{{{d_i}} \mathord{\left/
 {\vphantom {{{d_i}} 2}} \right.
 \kern-\nulldelimiterspace} 2}}} + 1)}}={e_{i,0}}(X)$ and ${\rho _{i,1}}{(X)^{k({q^{{{{d_i}} \mathord{\left/
 {\vphantom {{{d_i}} 2}} \right.
 \kern-\nulldelimiterspace} 2}}} + 1)}}={e_{i,1}}(X)$ if and only if
 $k({q^{{{{d_i}} \mathord{\left/
 {\vphantom {{{d_i}} 2}} \right.
 \kern-\nulldelimiterspace} 2}}} + 1) \equiv 0\;(\bmod \;{q^{{d_i}}} - 1)$ if and only if $k \equiv 0\;(\bmod \;{q^{{{{d_i}} \mathord{\left/
 {\vphantom {{{d_i}} 2}} \right.
 \kern-\nulldelimiterspace} 2}}} - 1)$ if and only if $k = ({q^{{{{d_i}} \mathord{\left/
 {\vphantom {{{d_i}} 2}} \right.
 \kern-\nulldelimiterspace} 2}}} - 1)m$ for $0 \le m \le {q^{{{{d_i}} \mathord{\left/
 {\vphantom {{{d_i}} 2}} \right.
 \kern-\nulldelimiterspace} 2}}}$,
since $d_i$ is even and $0\leq k \leq q^{d_i}-2$.\\
\indent
(iii) By part (ii), assume that ${\tau _{1, - 1}}({\mathcal {I}_{i,0}}) = {\mathcal {I}_{i,1}}$. As ${\tau _{q,1}}({\mathcal {I}_{i,0}}) = {\mathcal {I}_{i,1}}$, ${\tau _{q,-1}}({\mathcal {I}_{i,j}}) = {\mathcal {I}_{i,j}}$.  Recall that, ${\tau _{{q^2},1}}$ is the identity on $\mathcal {R}_n^{({q^2})}$,
\[\begin{array}{l}
\begin{array}{*{20}{l}}
{{{[a(X),b(X)]}_\Delta } = \gamma a(X){\tau _{{q^2}, - 1}}(b(X)) - \gamma {\tau _{q,1}}(a(X){\tau _{{q^2}, - 1}}(b(X)))}
\end{array}\\
\;\;\;\;\;\;\;\;\;\;\;\;\;\;\;\;\;\;\;\;\;\;\; = \begin{array}{*{20}{l}}
{\gamma a(X){\tau _{1, - 1}}(b(X)) - \gamma {\tau _{q,1}}(a(X){\tau _{1, - 1}}(b(X)))}
\end{array}.
\end{array}\]
It is not difficult to verify that, by Lemma 2 (v), ${[{e_{i,j}}(X),{e_{i,j}}(X)]_\Delta } = 0$ for $0\leq j\leq 1$.
So $\mathcal {C}_i$ is 1-dimensional over $\mathcal {K}_i$ with bases $\{ {e_{i,0}}(X) \}$ and $\{ {e_{i,1}}(X) \}$.
Since \[\begin{array}{l}
{[{e_{i,0}}(X) + {\rho _{i,1}}{(X)^k},{e_{i,0}}(X) + {\rho _{i,1}}{(X)^k}]_\Delta }\\
\;\;  = \gamma ({e_{i,0}}(X) + {\rho _{i,1}}{(X)^k})({e_{i,1}}{(X)^{{q^{{{{d_i}} \mathord{\left/
 {\vphantom {{{d_i}} 2}} \right.
 \kern-\nulldelimiterspace} 2}}}}} + {\rho _{i,0}}{(X)^{k{q^{{{{d_i}} \mathord{\left/
 {\vphantom {{{d_i}} 2}} \right.
 \kern-\nulldelimiterspace} 2}}}}})\\
\;\;\;\;\;\;  - \gamma {\tau _{q,1}}(({e_{i,0}}(X) + {\rho _{i,1}}{(X)^k})({e_{i,1}}{(X)^{{q^{{{{d_i}} \mathord{\left/
 {\vphantom {{{d_i}} 2}} \right.
 \kern-\nulldelimiterspace} 2}}}}} + {\rho _{i,0}}{(X)^{k{q^{{{{d_i}} \mathord{\left/
 {\vphantom {{{d_i}} 2}} \right.
 \kern-\nulldelimiterspace} 2}}}}}))\\
\;\;  = \gamma ({\rho _{i,0}}{(X)^{k{q^{{{{d_i}} \mathord{\left/
 {\vphantom {{{d_i}} 2}} \right.
 \kern-\nulldelimiterspace} 2}}}}} + {\rho _{i,1}}{(X)^k}) - \gamma ({\rho _{i,1}}{(X)^{k{q^{{{{d_i}} \mathord{\left/
 {\vphantom {{{d_i}} 2}} \right.
 \kern-\nulldelimiterspace} 2}}}}} + {\rho _{i,0}}{(X)^k})\\
\;\;  = \gamma ({\rho _{i,0}}{(X)^{k{q^{{{{d_i}} \mathord{\left/
 {\vphantom {{{d_i}} 2}} \right.
 \kern-\nulldelimiterspace} 2}}}}} - {\rho _{i,0}}{(X)^k}) - \gamma ({\rho _{i,1}}{(X)^{k{q^{{{{d_i}} \mathord{\left/
 {\vphantom {{{d_i}} 2}} \right.
 \kern-\nulldelimiterspace} 2}}}}} - {\rho _{i,1}}{(X)^k}),
\end{array}\]
${[{e_{i,0}}(X) + {\rho _{i,1}}{(X)^k},{e_{i,0}}(X) + {\rho _{i,1}}{(X)^k}]_\Delta } = 0$ if and only if ${\rho _{i,0}}{(X)^{k({q^{{{{d_i}} \mathord{\left/
 {\vphantom {{{d_i}} 2}} \right.
 \kern-\nulldelimiterspace} 2}}} - 1)}} = {e_{i,0}}(X)$ and ${\rho _{i,1}}{(X)^{k({q^{{{{d_i}} \mathord{\left/
 {\vphantom {{{d_i}} 2}} \right.
 \kern-\nulldelimiterspace} 2}}} - 1)}} = {e_{i,1}}(X)$ if and only if $k({q^{{{{d_i}} \mathord{\left/
 {\vphantom {{{d_i}} 2}} \right.
 \kern-\nulldelimiterspace} 2}}} - 1) \equiv 0\;(\bmod \;{q^{{d_i}}} - 1)$ if and only if
 $k \equiv 0\;(\bmod \;{q^{{{{d_i}} \mathord{\left/
 {\vphantom {{{d_i}} 2}} \right.
 \kern-\nulldelimiterspace} 2}}} + 1)$ if and only if $k = ({q^{{{{d_i}} \mathord{\left/
 {\vphantom {{{d_i}} 2}} \right.
 \kern-\nulldelimiterspace} 2}}} + 1)m$, where $0 \le m \le {q^{{{{d_i}} \mathord{\left/
 {\vphantom {{{d_i}} 2}} \right.
 \kern-\nulldelimiterspace} 2}}} - 2$ as $0 \le k \le {q^{d_i}-2}$.\\
 \indent
(iv) Now, suppose that $\mu(i)\neq i$ and $d_i$ is odd. By Lemmas 2 and 4, we have ${s_i} = \gcd (2,{d_i}) = 1$,
 ${\mathcal {J}_i} = {\mathcal {I}_{i,0}} \cong {\mathbb{F}_{{q^{2{D_i}}}}} = {\mathbb{F}_{{q^{2{d_i}}}}}$ and ${\mathcal {K}_i} = \{ c(X) \in {\mathcal {I}_{i,0}}|{\tau _{q,1}}(c(X)) = c(X)\}  \cong {\mathbb{F}_{{q^{{d_i}}}}}$.
 As ${\tau _{1, - 1}}$ is an isomorphism of $\mathcal {J}_i$ onto $\mathcal {J}_{\mu(i)}$, we have $d_{\mu(i)}=d_i$, ${s_{\mu(i)}} = \gcd (2,{d_{\mu(i)}}) = 1$, ${\mathcal {J}_{\mu(i)}} = {\mathcal {I}_{\mu(i),0}} \cong {\mathbb{F}_{{q^{2{d_i}}}}}$ and ${\mathcal {K}_{\mu(i)}} = \{ c(X) \in {\mathcal {I}_{\mu(i),0}}|{\tau _{q,1}}(c(X)) = c(X)\}  \cong {\mathbb{F}_{{q^{{d_i}}}}}$.
 ${\tau _{{q^2},1}}$ is the identity on $\mathcal {R}_n^{({q^2})}$ and hence ${\tau _{q,1}}:{\mathcal {J}_i} \to {\mathcal {J}_i}$ is the identity on $\mathcal {K}_i$ but not $\mathcal {J}_i$ by Lemma 4. Thus, ${\tau _{q,1}}$ is an automorphism of $\mathcal {J}_i$ of order 2 implying
 ${\tau _{q,1}}(c(X)) = c{(X)^{{q^{{{2{d_i}} \mathord{\left/
 {\vphantom {{2{d_i}} 2}} \right.
 \kern-\nulldelimiterspace} 2}}}}} = c{(X)^{{q^{{d_i}}}}}$ for all $c(X)\in \mathcal {J}_i$. Similarly, ${\tau _{q,1}}(c(X)) = c{(X)^{{q^{{d_i}}}}}$ for all $c(X)\in\mathcal {J}_{\mu(i)}$.\\
 \indent
 By Theorem 1, we have $$\mathcal {C}_{\mu (i)}^{(\Delta )} = \{ a(X) \in {\mathcal {J}_{\mu (i)}}|{[c(X),a(X)]_\Delta } = 0\;{\rm{for}}\;{\rm{all}}\;c(X) \in {\mathcal {C}_i}\} $$and ${\dim _{{\mathcal {K}_i}}}{\mathcal {C}_i} + {\dim _{{\mathcal {K}_{\mu (i)}}}}\mathcal {C}_{\mu (i)}^{(\Delta )} = t = 2.$ As ${\mathcal {C}_{\mu (i)}} \subseteq \mathcal {C}_{\mu (i)}^{(\Delta )}$ is required for cyclic $\Delta$-self-orthogonality, for each possible $\mathcal {C}_i$, we must find $\mathcal {C}_{\mu (i)}^{(\Delta )}$. It is easy to show that, as the dimensions add to 2, $\mathcal {C}_i=\{0\}$ implies $\mathcal {C}_{\mu (i)}^{(\Delta )}=\mathcal {J}_{\mu(i)}$ and $\mathcal {C}_i=\mathcal {J}_i$ implies $\mathcal {C}_{\mu (i)}^{(\Delta )}=\{0\}$. If $\mathcal {C}_i$ is 1-dimensional over $\mathcal {K}_i$, then $\mathcal {C}_{\mu (i)}^{(\Delta )}$ is 1-dimensional over $\mathcal {K}_{\mu(i)}$. Furthermore, if $\mathcal {C}_i$ is 1-dimensional over $\mathcal {K}_i$, then $\mathcal {C}_i$ has basis $\{ {\rho _{i,0}}{(X)^k}\} $ for some $k$. As $\mathcal {K}_i$ has primitive element ${\rho _{i,0}}{(X)^{{q^{{d_i}}} + 1}}$, $\mathcal {C}_i$ has ${q^{{d_i}}} + 1$ different bases $\{ {\rho _{i,0}}{(X)^k}\} $ for $0\leq k\leq q^{d_i}$. Now we find $k'$, where $0\leq k'\leq q^{d_i}$, such that ${[{\rho _{i,0}}{(X)^k},{\rho _{\mu (i),0}}{(X)^{k'}}]_\Delta } = 0$. Using ${\tau _{1, - 1}}({\rho _{i,0}}(X)) = {\rho _{\mu (i),0}}(X)$ and ${\tau _{q,1}}(c(X)) = c{(X)^{{q^{{d_i}}}}}$ for $c(X)\in \mathcal {J}_i$ or $\mathcal {J}_{\mu(i)}$, we have
 \[\begin{array}{l}
{[{\rho _{i,0}}{(X)^k},{\rho _{\mu (i),0}}{(X)^{k'}}]_\Delta }\\
\;\;  = \gamma {\rho _{i,0}}{(X)^k}{\tau _{1, - 1}}({\rho _{\mu (i),0}}{(X)^{k'}}) - \gamma {\tau _{q,1}}({\rho _{i,0}}{(X)^k}{\tau _{1, - 1}}({\rho _{\mu (i),0}}{(X)^{k'}}))\\
\;\;  = \gamma {\rho _{i,0}}{(X)^{k + k'}} - \gamma {\rho _{i,0}}{(X)^{(k + k'){q^{{d_i}}}}}\\
\;\;  = 0
\end{array}\]
if and only if ${\rho _{i,0}}{(X)^{(k + k')({q^{{d_i}}} - 1)}} = {e_{i,0}}(X)$ if and only if $(k + k')({q^{{d_i}}} - 1) \equiv 0\;(\bmod \;{q^{2{d_i}}} - 1)$ if and only if $k + k' \equiv 0\;(\bmod \;{q^{{d_i}}} + 1)$ if and only if $k'=k=0$ or $k'=q^{d_i}+1-k$ for $1\leq k\leq q^{d_i}$.\\
\indent
(v) Finally, assume that $\mu(i)\neq i$ and $d_i$ is even. From the proof of parts (ii)-(iv), we have $s_i=2$, ${\mathcal {J}_i} = {\mathcal {I}_{i,0}} \oplus {\mathcal {I}_{i,1}}$, ${\mathcal {I}_{i,0}} \cong {\mathcal {I}_{i,1}} \cong {\mathbb{F}_{{q^{{d_i}}}}}$, ${\mathcal {K}_i} = \{ c(X) + {\tau _{q,1}}(c(X))|c(X) \in {\mathcal {I}_{i,0}}\}  \cong {\mathbb{F}_{{q^{{d_i}}}}}$ and the $q^{d_i}+1$ 1-dimensional $\mathcal {K}_i$-subspaces of $\mathcal {J}_i$ have bases $\{ {e_{i,0}}(X)\} $, $\{ {e_{i,1}}(X)\} $ and $\{ {e_{i,0}}(X) + {\rho _{i,1}}{(X)^k}\} $, where $0\leq k \leq q^{d_i}-2$. In this case, ${\tau _{1, - 1}}({\rho _{i,j}}(X)) = {\rho _{\mu (i),j}}(X)$ for $0\leq j\leq 1$, ${\tau _{q,1}}({\rho _{i,0}}(X)) = {\rho _{i,1}}(X)$ and ${\tau _{q,1}}({\rho _{\mu (i),0}}(X)) = {\rho _{\mu (i),1}}(X)$.\\
\indent
If $\mathcal {C}_i=\{0\}$ or $\mathcal {C}_i=\mathcal {J}_i$, (a) and (b) follow as in the proof of part (iv). Now we consider the basis of $\mathcal {C}_{\mu (i)}^{(\Delta )}$ when $\mathcal {C}_i$ is 1-dimensional over $\mathcal {K}_i$. Then $\mathcal {C}_{\mu (i)}^{(\Delta )}$ will be 1-dimensional over $\mathcal {K}_{\mu(i)}$. As
\[{{{[a(X),b(X)]}_\Delta } = \begin{array}{*{20}{l}}
{\gamma a(X){\tau _{1, - 1}}(b(X)) - \gamma {\tau _{q,1}}(a(X){\tau _{1, - 1}}(b(X)))},
\end{array}}\]
we have \[\begin{array}{*{20}{l}}
\begin{array}{l}
{[{e_{i,0}}(X),{e_{\mu (i),0}}(X)]_\Delta }\\
\;\; = \gamma {e_{i,0}}(X){\tau _{1, - 1}}({e_{\mu (i),0}}(X)) - \gamma {\tau _{q,1}}({e_{i,0}}(X){\tau _{1, - 1}}({e_{\mu (i),0}}(X)))
\end{array}\\
{\;\;\; = \gamma {e_{i,0}}(X) - \gamma {e_{i,1}}(X) \ne 0.}
\end{array}\]
Similarly, ${[{e_{i,1}}(X),{e_{\mu (i),1}}(X)]_\Delta } \ne 0$. So neither $\{ {e_{i,0}}(X)\} $ nor $\{ {e_{i,1}}(X)\} $ is a basis of $\mathcal {C}_i$. Now we find $k'$, where $0\leq k'\leq q^{d_i}-2$, such that $${[{e_{i,0}}(X) + {\rho _{i,1}}{(X)^k},{e_{\mu (i),0}}(X) + {\rho _{\mu (i),1}}{(X)^{k'}}]_\Delta }= 0.$$ By ${\mathcal {I}_{i,0}}{\mathcal {I}_{i,1}} = \{ 0\} $, we have
\[\begin{array}{l}
{[{e_{i,0}}(X) + {\rho _{i,1}}{(X)^k},{e_{\mu (i),0}}(X) + {\rho _{\mu (i),1}}{(X)^{k'}}]_\Delta }\\
\;\; = \gamma ({e_{i,0}}(X) + {\rho _{i,1}}{(X)^k}){\tau _{1, - 1}}({e_{\mu (i),0}}(X) + {\rho _{\mu (i),1}}{(X)^{k'}})\\
\;\;\;\;\;\; - \gamma {\tau _{q,1}}(({e_{i,0}}(X) + {\rho _{i,1}}{(X)^k}){\tau _{1, - 1}}({e_{\mu (i),0}}(X) + {\rho _{\mu (i),1}}{(X)^{k'}}))\\
\;\; = \gamma ({e_{i,0}}(X) + {\rho _{i,1}}{(X)^k})({e_{i,0}}(X) + {\rho _{i,1}}{(X)^{k'}})\\
\;\;\;\;\;\; - \gamma ({e_{i,1}}(X) + {\rho _{i,0}}{(X)^k})({e_{i,1}}(X) + {\rho _{i,0}}{(X)^{k'}})\\
\;\; = \gamma ({e_{i,0}}(X) + {\rho _{i,1}}{(X)^{k + k'}}) - \gamma ({e_{i,1}}(X) + {\rho _{i,0}}{(X)^{k + k'}})\\
\;\; = \gamma ({e_{i,0}}(X) - {\rho _{i,0}}{(X)^{k + k'}}) - \gamma ({e_{i,1}}(X) - {\rho _{i,1}}{(X)^{k + k'}})\\
\;\; = 0
\end{array}\]
if and only if ${\rho _{i,0}}{(X)^{k + k'}} = {e_{i,0}}(X)$ and ${\rho _{i,1}}{(X)^{k + k'}} = {e_{i,1}}(X)$ if and only if $k + k' \equiv 0\;(\bmod \;{q^{{d_i}}} - 1)$ if and only if $k'=k=0$ or $k'=q^{d_i}-1-k$ for $1\leq k\leq q^{d_i}-2$.
\end{proof}
\par
When $t=2$, we can count the number of cyclic $\Delta$-self-orthogonal $\mathbb{F}_q$-linear $\mathbb{F}_{q^2}$-codes. Let $\mathfrak{J}$ be the fixed points of $\mu$ excluding $0$ and $i^\#$ (when $n$ is even). Let $\mathfrak{M}$ consist of one element from each of the transpositions in $\mu$. We give the number of cyclic $\Delta$-self-orthogonal $\mathbb{F}_q$-linear $\mathbb{F}_{q^2}$-codes in the following Theorem.
\par
\begin{theorem}
When $t=2$,  the number of cyclic $\Delta$-self-orthogonal $\mathbb{F}_q$-linear $\mathbb{F}_{q^2}$-codes of length $n$ is as follows, where $\gcd (n,q) = 1$.
\[a'\prod\nolimits_{i \in \mathfrak{J}} {({q^{{{{d_i}} \mathord{\left/
 {\vphantom {{{d_i}} 2}} \right.
 \kern-\nulldelimiterspace} 2}}} + 2)\prod\nolimits_{j \in \mathfrak{M}} {(3{q^{{d_j}}} + b'} )}, \]
 where $a'=2$ with $n$ odd, $a'=4$ with $n$ even, $b'=6$ if $d_i$ is odd and $b'=2$ if $d_i$ is even.
\end{theorem}
\begin{proof}
The trick of the proof is to count the number of subcodes $\mathcal {C}_i$ when $\mu(i)=i$ and subcodes pairs $(\mathcal {C}_i,\mathcal {C}_{\mu(i)})$ when $\mu(i)\neq i$ in Theorem 2. Let $M_i$ be the number of distinct $\mathcal {K}_i$-subspaces $\mathcal {C}_i$ of $\mathcal {J}_i$ satisfying ${\mathcal {C}_i} \subseteq \mathcal {C}_i^{(\Delta )}$, where $i \in \mathfrak{J}$. Let $M_j$ be the number of distinct pairs $(\mathcal {C}_j,\mathcal {C}_{\mu(j)})$, where $j\in \mathfrak{M}$, $\mathcal {C}_j$ is a $\mathcal {K}_j$-subspace of $\mathcal {J}_j$ and $\mathcal {C}_{\mu(j)}$ is a $\mathcal {K}_{\mu(j)}$-subspace of $\mathcal {J}_{\mu(j)}$, satisfying $\mathcal {C}_j\subseteq \mathcal {C}_j^{(\Delta )}$ and $\mathcal {C}_{\mu(j)}\subseteq \mathcal {C}_{\mu(j)}^{(\Delta )}$. Then, by Theorem 2, the number of distinct cyclic $\Delta$-self-orthogonal  $\mathbb{F}_q$-linear $\mathbb{F}_{q^2}$-codes of length $n$ is
\[\left\{ \begin{array}{l}
4\prod\nolimits_{i \in \mathfrak{J}} {{M_i}\prod\nolimits_{j \in \mathfrak{M}} {{M_j},\;\;\;{\rm{if}}\;n\;{\rm{is}}\;{\rm{even}},} } \\
2\prod\nolimits_{i \in \mathfrak{J}} {{M_i}\prod\nolimits_{j \in \mathfrak{M}} {{M_j},\;\;\;{\rm{if}}\;n\;{\rm{is}}\;{\rm{odd}}.} }
\end{array} \right.\]
It is straightforward to show that $M_i={{q^{{{{d_i}} \mathord{\left/
 {\vphantom {{{d_i}} 2}} \right.
 \kern-\nulldelimiterspace} 2}}} + 2}$ for $i \in \mathfrak{J}$, $M_j={3{q^{{d_j}}} + 6}$ for $d_i$ is odd and $j\in \mathfrak{M}$ and $M_j={3{q^{{d_j}}} + 2}$ for $d_i$ is even and $j\in \mathfrak{M}$.
 The proof is completed.
\end{proof}

%%%%%%%%%%%%%%%%%%%%%%%%%%%%%%%%%%%%%%%%%%%%%%%%%%%%%%%%%%%%%%%%%%%%%%%%%%%%%%%%%%%%%%%%%%
\subsection{Cyclic $\Delta$-self-dual $\mathbb{F}_q$-linear $\mathbb{F}_{q^2}$-codes} \label{}
\noindent
Using Theorem 2, we determine the bases of all cyclic $\Delta$-self-dual $\mathbb{F}_q$-linear $\mathbb{F}_{q^2}$-codes of length $n$ in the following Theorem. By Theorem 1, ${\dim _{{\mathcal {K}_i}}}{\mathcal {C}_i} + {\dim _{{\mathcal {K}_{\mu (i)}}}}{\mathcal {C}_{\mu (i)}} = t = 2$ for $0\leq i\leq s-1$. In particular, if $\mu(i)=i$, then ${\dim _{{\mathcal {K}_i}}}{\mathcal {C}_i} = 1$ and if $\mu(i)\neq i$, by Theorem 2, we have ${\mathcal {C}_{\mu (i)}} = \mathcal {C}_{\mu (i)}^{(\Delta )}$.
\begin{theorem}
Let $\mathcal {C}$ be a cyclic $\mathbb{F}_q$-linear $\mathbb{F}_{q^2}$-code of length $n$. If $\mathcal {C} = {\mathcal {C}_0} \oplus {\mathcal {C}_1} \oplus  \cdots  \oplus {\mathcal {C}_{s - 1}}$, where ${\mathcal {C}_i} = \mathcal {C} \cap {\mathcal {J}_i}$ for all $0\leq i\leq s-1$, then $\mathcal {C}$ is cyclic $\Delta$-self-dual when the following hold.\\
(i) When $i=0$ or, if $n$ is even, $i=i^\#$, then $\mathcal {C}_i$ is a 1-dimensional $\mathcal {K}_i$-subspace of $\mathcal {J}_i$ with basis $\{ {\rho _{i,0}}{(X)^k}\} $, where $k=0$ when $q$ is even and $k=\frac{{q + 1}}{2}$ when $q$ is odd.\\
(ii) When $i\neq 0$, $i\neq i^\#$, $\mu(i)=i$ and ${\tau _{1, - 1}}({\mathcal {I}_{i,0}}) = {\mathcal {I}_{i,0}}$, then $\mathcal {C}_i$ is 1-dimensional over $\mathcal {K}_i$ with basis $\{ {e_{i,0}}(X) + {\rho _{i,1}}{(X)^k}\} $, where $k = ({q^{{{{d_i}} \mathord{\left/
 {\vphantom {{{d_i}} 2}} \right.
 \kern-\nulldelimiterspace} 2}}} - 1)m$ for $0\leq m\leq {q^{{{{d_i}} \mathord{\left/
 {\vphantom {{{d_i}} 2}} \right.
 \kern-\nulldelimiterspace} 2}}}$.\\
(iii) When $i\neq 0$, $i\neq i^\#$, $\mu(i)=i$ and ${\tau _{1, - 1}}({\mathcal {I}_{i,0}}) = {\mathcal {I}_{i,1}}$, then $\mathcal {C}_i$ is 1-dimensional over $\mathcal {K}_i$ with bases $\{ {e_{i,0}}(X) \}$, $\{ {e_{i,1}}(X) \}$ and $\{ {e_{i,0}}(X) + {\rho _{i,1}}{(X)^k}\} $, where $k = ({q^{{{{d_i}} \mathord{\left/
 {\vphantom {{{d_i}} 2}} \right.
 \kern-\nulldelimiterspace} 2}}} + 1)m$ for $0 \le m \le {q^{{{{d_i}} \mathord{\left/
 {\vphantom {{{d_i}} 2}} \right.
 \kern-\nulldelimiterspace} 2}}} - 2$.\\
 (iv) When $\mu(i)\neq i$ and $d_i$ is odd,\\
 \indent
 (a) if $\mathcal {C}_i=\{0\}$, then ${\mathcal {C}_{\mu (i)}} = {\mathcal {J}_{\mu (i)}}$;\\
\indent
 (b) if $\mathcal {C}_i=\mathcal {J}_i$, then ${\mathcal {C}_{\mu (i)}}=\{0\} $;\\
\indent
 (c) if $\mathcal {C}_i$ is 1-dimensional over $\mathcal {K}_i$ with basis $\{{\rho _{i,0}}{(X)^k}\} $, where $0 \le k \le {q^{{d_i}}}$, then ${\mathcal {C}_{\mu (i)}}$ is 1-dimensional over $\mathcal {K}_{\mu(i)}$ with basis $\{ {\rho _{\mu (i),0}}{(X)^{k'}}\} $, where $k'=k=0$ or $k'=q^{d_i}+1-k$ for $1 \le k \le {q^{{d_i}}}$.\\
 (v)  When $\mu(i)\neq i$ and $d_i$ is even,\\
 \indent
 (a) if $\mathcal {C}_i=\{0\}$, then ${\mathcal {C}_{\mu (i)}} = {\mathcal {J}_{\mu (i)}}$;\\
\indent
 (b) if $\mathcal {C}_i=\mathcal {J}_i$, then ${\mathcal {C}_{\mu (i)}}=\{0\} $;\\
\indent
 (c) if $\mathcal {C}_i$ is 1-dimensional over $\mathcal {K}_i$ with basis $\{ a(X)\} $, then ${\mathcal {C}_{\mu (i)}}$ is 1-dimensional over $\mathcal {K}_{\mu(i)}$ with basis $\{ a'(X)\} $, where $a'(X) = {e_{\mu (i),0}}(X) + {\rho _{\mu (i),1}}{(X)^{k'}}$ when $a(X) = {e_{i,0}}(X) + {\rho _{i,1}}{(X)^k}$ with
  $k'=k=0$ or $k'=q^{d_i}-1-k$ for $1 \le k \le {q^{{d_i}}-2}$.
\end{theorem}
\begin{proof}
The result follows immediately by Theorem 2.
\end{proof}
\par
When $t=2$, similarly, we give the number of cyclic $\Delta$-self-dual $\mathbb{F}_q$-linear $\mathbb{F}_{q^2}$-codes in the following Theorem.
\begin{theorem}
When $t=2$, the number of cyclic $\Delta$-self-dual $\mathbb{F}_q$-linear $\mathbb{F}_{q^2}$-codes of length $n$ is as follows, where $\gcd (n,q) = 1$.
\[\prod\nolimits_{i \in \mathfrak{J}} {({q^{{{{d_i}} \mathord{\left/
 {\vphantom {{{d_i}} 2}} \right.
 \kern-\nulldelimiterspace} 2}}} + 1)\prod\nolimits_{j \in \mathfrak{M}} {({q^{{d_j}}} + b')} }, \]
 where $ b'=3$ if $d_i$ is odd and $ b'=1$ if $d_i$ is even.
\end{theorem}
\begin{proof}
The proof of this result is quite similar to Theorem 3, and so is omitted.
\end{proof}

%%%%%%%%%%%%%%%%%%%%%%%%%%%%%%%%%%%%%%%%%%%%%%%%%%%%%%%%%%%%%%%%%%%%%%%%%%%%%%%%%%%%%%%%%%%%%%
\section{An example} \label{}
\noindent
In this case, we have $q=p=3$, $t=2$ and $n=7$. Then $t=2^am$ and $Q=2^{a-1}=1$, where $a=1, m=1$. Let $\mathbb{F}_{3^2}={{{\mathbb{F}_3}[X]} \mathord{\left/
 {\vphantom {{{F_3}[X]} {\left\langle {{X^2} + 1} \right\rangle }}} \right.
 \kern-\nulldelimiterspace} {\left\langle X^2+2X+2 \right\rangle }}$, where $X^2+2X+2$ is a primitive polynomial over $\mathbb{F}_3[X]$. Let $\omega$ be a root of $X^2+2X+2$, it is trivial to see that $\omega$ is a primitive element of $\mathbb{F}_{3^2}$ with ${\rm{ord}}(\omega ) = {3^2} - 1 = 8$. Thus $\mathcal {R}_{7}^{(3)} = {{{\mathbb{F}_3}[X]} \mathord{\left/
 {\vphantom {{{\mathbb{F}_3}[X]} {\left\langle {{X^{7}} - 1} \right\rangle }}} \right.
 \kern-\nulldelimiterspace} {\left\langle {{X^{7}} - 1} \right\rangle }}$ and $\mathcal {R}_{7}^{(3^2)} = {{{\mathbb{F}_{3^2}}[X]} \mathord{\left/
 {\vphantom {{{\mathbb{F}_3}[X]} {\left\langle {{X^{7}} - 1} \right\rangle }}} \right.
 \kern-\nulldelimiterspace} {\left\langle {{X^{7}} - 1} \right\rangle }}$. All distinct cyclic $\Delta$-self-orthogonal and cyclic $\Delta$-self-dual $\mathbb{F}_3$-linear $\mathbb{F}_{3^2}$-codes of length $7$ and their enumeration can be given in the following Steps.
\par
\textbf{Step 1} Factorize ${X^{7}} - 1 = \prod\nolimits_{i = 0}^1 {{m_i}} (X)$, where $m_0(X)=X-1=X+2$ and $m_1(X)=X^6+X^5+X^4+X^3+X^2+X+1$ are monic irreducible polynomials over $\mathbb{F}_3[X]$. Let $\mathbb{F}_{3^6}={{{\mathbb{F}_3}[X]} \mathord{\left/
 {\vphantom {{{F_3}[X]} {\left\langle {X^6 + 2X^4 + X^2 + 2X + 2} \right\rangle }}} \right.
 \kern-\nulldelimiterspace} {\left\langle X^6 + 2X^4 + X^2 + 2X + 2 \right\rangle }}$, where $X^6 + 2X^4 + X^2 + 2X + 2$ is a primitive polynomial  over $\mathbb{F}_3[X]$, and $\eta$ be a root of $X^6 + 2X^4 + X^2 + 2X + 2 $, then we have ${\rm{ord}}(\eta) = {3^6} - 1 = 728$.
Thus $\eta'=\eta^{104}$ is a primitive $7$th root of unity over $\mathbb{F}_{3^6}$.
\par
According to Lemma 2 (i), as $C_0^{(3)} = \{ 0\} $ and $C_1^{(3)} = \{ 1,2,3,4,5,6\} $, it is not difficult to find that ${l_0} = 0,\;{l_1} = 1$, ${d_0} = \deg {m_0}(X) =|C_{{l_0}}^{(3)}| = 1$ and ${d_1} = \deg {m_1}(X) = |C_{{l_1}}^{(3)}| = 6$. Furthermore, we have ${m_0}(X) \leftrightarrow C_{{l_0}}^{(3)} = C_0^{(3)}$ and ${m_1}(X) \leftrightarrow C_{{l_1}}^{(3)} = C_1^{(3)}$. That is,
\[{m_0}(X) = X - 1 = X - {({\eta ^{104}})^0} \]
and
\[\begin{array}{*{20}{l}}
{{m_1}(X) = {X^6} + {X^5} + {X^4} + {X^3} + {X^2} + X + 1}\\
\begin{array}{l}
\;\;\;\;\;\;\;\;\;\; \;= (X - {({\eta ^{104}})^1})(X - {({\eta ^{104}})^2})(X - {({\eta ^{104}})^3})(X - {({\eta ^{104}})^4})\\
\;\;\;\;\;\;\;\;\;\;\;\;\;\;\; \cdot (X - {({\eta ^{104}})^5})(X - {({\eta ^{104}})^6})
\end{array}\\
{\;\;\;\;\;\;\;\;\;\; \;\;= (X - {\eta ^{104}})(X - {\eta ^{208}})(X - {\eta ^{312}})(X - {\eta ^{416}})(X - {\eta ^{520}})(X - {\eta ^{624}}).}
\end{array}\]
\par
\textbf{Step 2} By computer systems Maple and Magma (http://magma.maths.us\\
yd.edu.au/calc/), it is easy to verify that ${m_0}(X) = X - 1 = {M_{0,0}}(X)$ and ${m_1}(X) ={X^6} + {X^5} + {X^4} + {X^3} + {X^2} + X + 1= {M_{1,0}}(X){M_{1,1}}(X)$, where ${M_{1,0}}(X) = X^3 + \omega X^2 + \omega^7X + 2$, ${M_{1,1}}(X) = X^3 + \omega^3X^2 + \omega^5X + 2$  and ${M_{0,0}}(X)$, ${M_{1,0}}(X)$ ,${M_{1,1}}(X)$ are monic irreducible polynomials over $\mathbb{F}_{3^2}[X]$. So, by Lemma 2 (ii) and (iii), we have $s_0=1$, $s_1=2$ and ${X^{7}} - 1 = \prod\nolimits_{i = 0}^1 {\prod\nolimits_{j = 0}^{{s_i} - 1} {{M_{i,j}}} (X)} .$
\par
Since ${M_{i,j}}(X) \leftrightarrow C_{{l_i}{3^j}}^{({3^2})}$, ${M_{0,0}}(X) \leftrightarrow C_{{l_0}{3^0}}^{({3^2})} = C_0^{(9)} = \{ 0\} ,$ ${M_{1,0}}(X) \leftrightarrow C_{{l_1}{3^0}}^{({3^2})} = C_1^{(9)} = \{ 1,2,4\} $ and ${M_{1,1}}(X) \leftrightarrow C_{{l_1}{3^1}}^{({3^2})} = C_3^{(9)} = \{ 3,5,6\} $. That is,
\[{M_{0,0}}(X) = X - 1 =  X - {({\eta ^{104}})^0},\]
\[\begin{array}{*{20}{l}}
\begin{array}{l}
{M_{1,0}}(X)\\
\;\; = {X^3} + \omega {X^2} + {\omega ^7}X + 2 = (X - {({\eta ^{104}})^1})(X - {({\eta ^{104}})^2})(X - {({\eta ^{104}})^4})
\end{array}\\
{\;\;\; = (X - {\eta ^{104}})(X - {\eta ^{208}})(X - {\eta ^{416}})}
\end{array}\]
and
\[\begin{array}{*{20}{l}}
\begin{array}{l}
{M_{1,1}}(X)\\
\;\; = {X^3} + {\omega ^3}{X^2} + {\omega ^5}X + 2 = (X - {({\eta ^{104}})^3})(X - {({\eta ^{104}})^5})(X - {({\eta ^{104}})^6})
\end{array}\\
{\;\;\; = (X - {\eta ^{312}})(X - {\eta ^{520}})(X - {\eta ^{624}}).}
\end{array}\]
\par
\textbf{Step 3 } By Lemma 2 (iv) and (v), we have $\mathcal {R}_{7}^{(3)} =\mathcal {K}_0\oplus \mathcal {K}_1$, where $\mathcal {K}_0$ is the ideal of $\mathcal {R}_{7}^{(3)}$ generated by ${{\hat m}_0}(X) = {{({X^{7}} - 1)} \mathord{\left/
 {\vphantom {{({X^{7}} - 1)} {{m_0}(X)}}} \right.
 \kern-\nulldelimiterspace} {{m_0}(X)}} =X^6+ X^5+X^4+ X^3+X^2+X+1$ denoted by $\mathcal {K}_0=\langle X^6+ X^5+X^4+ X^3+X^2+X+1 \rangle \cong \mathbb{F}_{3^{d_0}}=\mathbb{F}_3$. Similarly, we have $\mathcal {K}_1=\langle X+2 \rangle \cong \mathbb{F}_{3^{d_1}}=\mathbb{F}_{3^6}$.
 \par
 As ${D_i} = {{{d_i}} \mathord{\left/
 {\vphantom {{{d_i}} {{s_i}}}} \right.
 \kern-\nulldelimiterspace} {{s_i}}}$ for $0\leq i\leq 1$, $D_0=1$ and $D_1=3$.
 Then, similarly, we have $\mathcal {R}_{7}^{(3^2)} =\mathcal {I}_{0,0}\oplus \mathcal {I}_{1,0}\oplus \mathcal {I}_{1,1}$, where $\mathcal {I}_{0,0}=\langle{{({X^7} - 1)} \mathord{\left/
 {\vphantom {{({X^7} - 1)} {{M_{0,0}}(X)}}} \right.
 \kern-\nulldelimiterspace} {{M_{0,0}}(X)}}\rangle=\langle X^6+ X^5+X^4+ X^3+X^2+X+1 \rangle \cong \mathbb{F}_{3^{2D_0}}=\mathbb{F}_{3^2}$, $\mathcal {I}_{1,0}=\langle X^4 + \omega^5 X^3 + 2 X^2 + \omega^7 X + 1 \rangle \cong \mathbb{F}_{3^{2D_1}}=\mathbb{F}_{3^6}$ and $\mathcal {I}_{1,1}=\langle X^4 + \omega^7 X^3 + 2 X^2 + \omega^5 X + 1 \rangle \cong \mathbb{F}_{3^{2D_1}}=\mathbb{F}_{3^6}$. In addition, we have $\mathcal {J}_0=\mathcal {I}_{0,0}$ and $\mathcal {J}_1=\mathcal {I}_{1,0}\oplus \mathcal {I}_{1,1}$.
 \par
\textbf{ Step 4} By Wan \cite{Wan} and computer systems Maple and Magma, we can obtain non-zero idempotents $e_{i,j}(X)$ of $\mathcal {I}_{i,j}$, where $0\leq i\leq 1$ and $0\leq j\leq s_i-1$.
That is, $$e_{0,0}(X)=X^6 + X^5 + X^4 + X^3 + X^2 + X + 1,$$
$$e_{1,0}(X)=\omega^5 X^6 + \omega^5 X^5 + \omega^7 X^4 + \omega^5 X^3 + \omega^7 X^2 + \omega^7 X$$
and
$$e_{1,1}(X)=\omega^7 X^6 + \omega^7 X^5 + \omega^5 X^4 + \omega^7 X^3 + \omega^5 X^2 + \omega^5 X.$$
It is easy to show that $e_{i,j}(X)$ is the identity of $\mathcal {I}_{i,j}$ for $0\leq i\leq 1$ and $0\leq j\leq s_i-1$.
\par
It is easy to verify that
 $\eta$ and $\eta^{243}$ are primitive elements of finite field $\mathbb{F}_{3^6}$.
As $\mathcal {I}_{0,0}\cong \mathbb{F}_{3^2}$, $\mathcal {I}_{1,0}\cong \mathbb{F}_{3^6}$ and $\mathcal {I}_{1,1}\cong \mathbb{F}_{3^6}$, we choose
$${\rho _{0,0}}(X) = \omega \cdot {e_{0,0}}(X) = \omega X^6 + \omega X^5 +\omega X^4 +\omega X^3 +\omega X^2 +\omega X + \omega,$$
\[\begin{array}{*{20}{l}}
\begin{array}{l}
{\rho _{1,0}}(X)\\
\;\; = {\eta ^{243}}\cdot{e_{1,0}}(X)
\end{array}\\
{\;\;\; = {\eta ^{243}}{\omega ^5}{X^6} + {\eta ^{243}}{\omega ^5}{X^5} + {\eta ^{243}}{\omega ^7}{X^4} + {\eta ^{243}}{\omega ^5}{X^3} + {\eta ^{243}}{\omega ^7}{X^2} + {\eta ^{243}}{\omega ^7}X}
\end{array}\]
and
$${\rho _{1,1}}(X) =\eta \cdot{e_{1,1}}(X)= \eta \omega^7 X^6 +\eta  \omega^7 X^5 + \eta \omega^5 X^4 +\eta  \omega^7 X^3 +\eta  \omega^5 X^2 +\eta  \omega^5 X$$
are primitive elements of $\mathcal {I}_{i,j}$ satisfying ${\tau _{{3^j},1}}({\rho _{i,0}}(X)) = {\rho _{i,j}}(X)$ for $0\leq i\leq 1$ and $0\leq j\leq s_i-1$.
\par
\textbf{Step 5} It is easy to show that $C_{ - {l_1}}^{(3)} = C_{ - 1}^{(3)} = \{ 1,2,3,4,5,6\}  = C_1^{(3)} = C_{{l_1}}^{(3)} = C_{{l_{\mu (1)}}}^{(3)}$, then $\mu(1)=1$. By Maple and Magma, we verify that ${\tau _{1, - 1}}({\mathcal {I}_{1,0}}) = {\mathcal {I}_{1,1}}$.
\par
Let $\mathcal {C}$ be a cyclic $\mathbb{F}_3$-linear $\mathbb{F}_{3^2}$-code of length $7$ and ${\mathcal {C}^{{ \bot _\Delta }}}$ the dual code of $\mathcal {C}$. We write $\mathcal {C} = {\mathcal {C}_0} \oplus {\mathcal {C}_1}$ and ${\mathcal {C}^{{ \bot _\Delta }}} = \mathcal {C}_0^{(\Delta )} \oplus \mathcal {C}_1^{(\Delta )} $, where ${\mathcal {C}_i} = \mathcal {C} \cap {\mathcal {J}_i}$ and $\mathcal {C}_i^{(\Delta )} = {\mathcal {C}^{{ \bot _\Delta }}} \cap {\mathcal {J}_i}$ for all $0\leq i\leq 1$. According to Theorem 2, $\mathcal {C}$ is cyclic $\Delta$-self-orthogonal if and only if for each $i$ $(0\leq i\leq 1)$, the following assertions hold.\\
\indent
(i) When $i=0$, then\\
\indent
(a) $\mathcal {C}_ 0=\{0\}$, or\\
\indent
(b) $\mathcal {C}_0$ is 1-dimensional $\mathcal {K}_0$-subspace of $\mathcal {J}_0$ with basis $\{ {\rho _{0,0}}{(X)^2}\}= \{\omega^2 X^6 + \omega^2 X^5 +\omega^2 X^4 +\omega^2 X^3 +\omega^2 X^2 +\omega^2 X + \omega^2\}$.\\
\indent
(ii) When $i\neq 0$, we have $\mu(1)=1$ and ${\tau _{1, - 1}}({\mathcal {I}_{1,0}}) = {\mathcal {I}_{1,1}}$. Then\\
\indent
(a) $\mathcal {C}_1=\{0\}$, or\\
\indent
(b) $\mathcal {C}_1$ is 1-dimensional over $\mathcal {K}_1$ with bases $\{ {e_{1,0}}(X)\}  = \{ \omega^5 X^6 + \omega^5 X^5 + \omega^7 X^4 + \omega^5 X^3 + \omega^7 X^2 + \omega^7 X\}$, $\{ {e_{1,1}}(X)\}  = \{ \omega^7 X^6 + \omega^7 X^5 + \omega^5 X^4 + \omega^7 X^3 + \omega^5 X^2 + \omega^5 X\}$ and
 $\{e_{1,0}(X)+{{\rho _{1,1}}{(X)}^k}\},$ where $k=28m$ for $0\leq m\leq 25$. \\
 \indent
According to Theorem 3, the number of cyclic $\Delta$-self-orthogonal $\mathbb{F}_3$-linear $\mathbb{F}_{3^2}$-codes of length $7$ is
$2 \cdot ({3^{{{{d_1}} \mathord{\left/
 {\vphantom {{{d_1}} 2}} \right.
 \kern-\nulldelimiterspace} 2}}} + 2) = 58$.
 \par
\textbf{ Step 6} By Theorem 4, $\mathcal {C}$ is cyclic $\Delta$-self-dual if and only if for each $i$ $(0\leq i\leq 1)$, the following assertions hold.\\
\indent
(i) When $i=0$, then $\mathcal {C}_0$ is 1-dimensional $\mathcal {K}_0$-subspace of $\mathcal {J}_0$ with basis $\{ {\rho _{0,0}}{(X)^2}\}= \{\omega^2 X^6 + \omega^2 X^5 +\omega^2 X^4 +\omega^2 X^3 +\omega^2 X^2 +\omega^2 X + \omega^2\}$.\\
\indent
(ii)  When $i\neq 0$, we have $\mu(1)=1$ and ${\tau _{1, - 1}}({\mathcal {I}_{1,0}}) = {\mathcal {I}_{1,1}}$. Then $\mathcal {C}_1$ is 1-dimensional over $\mathcal {K}_1$ with bases $\{ {e_{1,0}}(X)\}  = \{ \omega^5 X^6 + \omega^5 X^5 + \omega^7 X^4 + \omega^5 X^3 + \omega^7 X^2 + \omega^7 X\}$, $\{ {e_{1,1}}(X)\}  = \{ \omega^7 X^6 + \omega^7 X^5 + \omega^5 X^4 + \omega^7 X^3 + \omega^5 X^2 + \omega^5 X\}$ and
 $\{e_{1,0}(X)+{{\rho _{1,1}}{(X)}^k}\},$ where $k=28m$ for $0\leq m\leq 25$.\\
 \indent
According to Theorem 5, the number of cyclic $\Delta$-self-dual $\mathbb{F}_3$-linear $\mathbb{F}_{3^2}$-codes of length $7$ is
${3^{{{{d_1}} \mathord{\left/
 {\vphantom {{{d_1}} 2}} \right.
 \kern-\nulldelimiterspace} 2}}} + 1 = 28$.

\par
\textbf{ A good code} By Step 5, we have that $\mathcal {C}=\langle {e_{1,0}}(X) \rangle=\langle  \omega^5 X^6 + \omega^5 X^5 + \omega^7 X^4 + \omega^5 X^3 + \omega^7 X^2 + \omega^7 X \rangle$ is a cyclic $\Delta$-self-orthogonal $\mathbb{F}_3$-linear $\mathbb{F}_{3^2}$-code of length $7$. Using Magma, it is not difficult to show that $\mathcal {C}$ is a $(7,(3^2)^3,5)$ $\mathbb{F}_3$-linear $\mathbb{F}_{3^2}$-code with generator matrix \[\left( {\begin{array}{*{20}{c}}
0&w^7& w^7&w^5& w^7&w^5& w^5\\
w^5&0&w^7& w^7&w^5& w^7&w^5 \\
w^5&w^5&0&w^7& w^7&w^5& w^7 \\
w^7&w^5&w^5&0&w^7& w^7&w^5  \\
w^5&w^7&w^5&w^5&0&w^7& w^7  \\
 w^7&w^5&w^7&w^5&w^5&0&w^7
\end{array}} \right).\]
 Furthermore, $\mathcal {C}$ is a good code which has the same parameters with the best known linear code $[7,3,5]$ over $\mathbb{F}_{3^2}$ according to the online database \cite{Markus}.
\par

At the end of this example, with the similar manner above, we construct several good cyclic $\Delta$-self-orthogonal $\mathbb{F}_q$-linear $\mathbb{F}_{q^2}$-codes of length $n$ in Table 1. These codes have the same parameters with the best known linear codes over $\mathbb{F}_{q^2}$ given in online database \cite{Markus} or MDS codes over $\mathbb{F}_{q^2}$ . In Table 1, $n$ is the length of $\mathcal {C}$, $(q^2)^k$ is the cardinality of $\mathcal {C}$, $d$ is the minimum Hamming diatance of $\mathcal {C}$ and $\alpha_i$ are given in Appendix A, where $1\leq i\leq 15$.
\par
\begin{table}[h]
\small
\caption{Some good  $\mathbb{F}_q$-linear $\mathbb{F}_{q^2}$-codes}
\label{tab:1}       % Give a unique label
\begin{tabular}{lcl}
\hline\noalign{\smallskip}
$\{q,n\}$&The basis of $\mathcal {C}$ & $(n,(q^2)^k,d)$ \\
\noalign{\smallskip}\hline\noalign{\smallskip}
$\{2,11\}$&$\alpha_1$ & $(11,(2^2)^5,6)$ \\
$\{2,19\}$&$\alpha_2$ & $(19,(2^2)^9,8)$ \\
$\{3,7\}$&$\alpha_3$ & $(7,(3^2)^3,5)$ \\
$\{3,19\}$&$\alpha_4$ & $(19,(3^2)^9,10)$ \\
$\{5,7\}$&$\alpha_5$ & $(7,(5^2)^3,5)$ \\
$\{5,23\}$&$\alpha_6$ & $(23,(5^2)^{11},12)$ \\
$\{7,11\}$&$\alpha_7$ & $(11,(7^2)^{5},7)$ \\
$\{7,23\}$&$\alpha_8$ & $(23,(7^2)^{11},12)$ \\
$\{11,23\}$&$\alpha_9$ & $(23,(11^2)^{11},12)$ \\
$\{13,11\}$&$\alpha_{10}$ & $(11,(13^2)^{5},7)$ \\
$\{13,19\}$&$\alpha_{11}$ & $(19,(13^2)^{9},11)$ \\
$\{17,7\}$&$\alpha_{12}$ & $(7,(17^2)^{3},5)$ \\
$\{17,11\}$&$\alpha_{13}$ & $(11,(17^2)^{5},7)$ \\
$\{19,7\}$&$\alpha_{14}$ & $(7,(19^2)^{3},5)$ \\
$\{19,11\}$&$\alpha_{15}$ & $(11,(19^2)^{5},7)$ \\
\noalign{\smallskip}\hline
\end{tabular}
\end{table}

%%%%%%%%%%%%%%%%%%%%%%%%%%%%%%%%%%%%%%%%%%%%%%%%%%%%%%%%%%%%%%%%%%%%%%%%%%%%%%%%%%%%%%%%%%%%%%
\section{Conclusions} \label{}
\noindent A new trace bilinear form on $\mathbb{F}_{{q^t}}^n$ which is called $\Delta$-bilinear form is given, where $n$ is a positive integer coprime to $q$. Then according to this new trace bilinear form, bases and enumeration of cyclic $\Delta$-self-orthogonal and cyclic $\Delta$-self-dual $\mathbb{F}_q$-linear $\mathbb{F}_{q^t}$-codes are investigated when $t=2$. Finally, we describe a program to construct cyclic $\Delta$-self-orthogonal and cyclic $\Delta$-self-dual $\mathbb{F}_3$-linear $\mathbb{F}_{3^2}$-codes of length $7$ and obtain some good $\mathbb{F}_q$-linear $\mathbb{F}_{q^2}$-codes.

\vskip 3mm \noindent {\bf Acknowledgments} This research is supported by the National Key Basic Research Program of China (Grant No. 2013CB834204), and the National Natural Science Foundation of China (Nos.61171082,61571243).

\vskip 3mm \noindent{\bf Appendix A}

\vskip 3mm
In Table 1, \\
\indent
$\alpha_1={\omega_1}X^{10} + {\omega_1}^2 X^9 + {\omega_1} X^8 + {\omega_1} X^7 + {\omega_1} X^6 + {\omega_1}^2 X^5 + {\omega_1}^2 X^4 + {\omega_1}^2 X^3 + {\omega_1} X^2 +{\omega_1}^2 X + 1,$\\
\par
$\alpha_2={\omega_1} X^{18} + {\omega_1}^2 X^{17} + {\omega_1}^2 X^{16} + {\omega_1} X^{15} + {\omega_1} X^{14} + {\omega_1} X^{13} + {\omega_1} X^{12} +{\omega_1}^2 X^{11} +{\omega_1} X^{10} + {\omega_1}^2 X^9 +{\omega_1} X^8 + {\omega_1}^2 X^7 + {\omega_1}^2 X^6 +{\omega_1}^2 X^5 + {\omega_1}^2 X^4 + {\omega_1} X^3 + {\omega_1} X^2 + {\omega_1}^2 X + 1,$\\
\par
$\alpha_3=\omega^5 X^6 + \omega^5 X^5 + \omega^7 X^4 + \omega^5 X^3 + \omega^7 X^2 + \omega^7 X,$\\
\par
$\alpha_4=\omega^5 X^{18} + \omega^7 X^{17} + \omega^7 X^{16} + \omega^5 X^{15} + \omega^5 X^{14} + \omega^5 X^{13} + \omega^5 X^{12} +
    \omega^7 X^{11} + \omega^5 X^{10} + \omega^7 X^9 + \omega^5 X^8 + \omega^7 X^7 + \omega^7 X^6 + \omega^7 X^5 +
    \omega^7 X^4 + \omega^5 X^3 + \omega^5 X^2 + \omega^7 X,$\\
    \par
$\alpha_5={\omega_2}^7 X^6 + {\omega_2}^7 X^5 + {\omega_2}^{11} X^4 + {\omega_2}^7 X^3 + {\omega_2}^{11} X^2 + {\omega_2}^{11} X + 4,$\\
\par
$\alpha_6={\omega_2}^{22} X^{22} + {\omega_2}^{22} X^{21} + {\omega_2}^{22} X^{20} + {\omega_2}^{22} X^{19} + {\omega_2}^{14} X^{18} + {\omega_2}^{22} X^{17} + {\omega_2}^{14} X^{16} + {\omega_2}^{22} X^{15} + {\omega_2}^{22} X^{14} + {\omega_2}^{14} X^{13} + {\omega_2}^{14} X^{12} + {\omega_2}^{22} X^{11} +{\omega_2}^{22} X^{10} + {\omega_2}^{14} X^9 + {\omega_2}^{14} X^8 + {\omega_2}^{22} X^7 + {\omega_2}^{14} X^6 + {\omega_2}^{22} X^5 + {\omega_2}^{14} X^4
    + {\omega_2}^{14} X^3 +{\omega_2}^{14} X^2 + {\omega_2}^{14} X + 2,$\\
    \par
$\alpha_7={\omega_3}^{41} X^{10} + {\omega_3}^{47} X^9 + {\omega_3}^{41} X^8 + {\omega_3}^{41} X^7 + {\omega_3}^{41} X^6 + {\omega_3}^{47} X^5 + {\omega_3}^{47} X^4 +{\omega_3}^{47} X^3 + {\omega_3}^{41} X^2 + {\omega_3}^{47} X + 3,$\\
\par
$\alpha_8={\omega_3}^{35} X^{22} + {\omega_3}^{35} X^{21} + {\omega_3}^{35} X^{20} + {\omega_3}^{35} X^{19} + {\omega_3}^5 X^{18} + {\omega_3}^{35} X^{17} + {\omega_3}^5 X^{16} + {\omega_3}^{35} X^{15} + {\omega_3}^{35} X^{14} + {\omega_3}^5 X^{13} + {\omega_3}^5 X^{12} + {\omega_3}^{35} X^{11} + {\omega_3}^{35} X^{10} +
    {\omega_3}^5 X^9 + {\omega_3}^5 X^8 + {\omega_3}^{35} X^7 + {\omega_3}^5 X^6 + {\omega_3}^{35} X^5 + {\omega_3}^5 X^4 + {\omega_3}^5 X^3 +
    {\omega_3}^5 X^2 + {\omega_3}^5 X + 2,$\\
\par
$\alpha_9={\omega_4}^{99} X^{22} + {\omega_4}^{99} X^{21} + {\omega_4}^{99} X^{20} + {\omega_4}^{99} X^{19} + {\omega_4}^9 X^{18} + {\omega_4}^{99} X^{17} + {\omega_4}^9 X^{16}+ {\omega_4}^{99} X^{15} + {\omega_4}^{99} X^{14} + {\omega_4}^9 X^{13} + {\omega_4}^9 X^{12} + {\omega_4}^{99} X^{11} + {\omega_4}^{99} X^{10} +  {\omega_4}^9 X^9 + {\omega_4}^9 X^8 + {\omega_4}^{99} X^7 + {\omega_4}^9 X^6 + {\omega_4}^{99} X^5 +{\omega_4}^9 X^4 + {\omega_4}^9 X^3 +{\omega_4}^9 X^2 + {\omega_4}^9 X,$\\
\par
$\alpha_{10}={\omega_5}^{158} X^{10} + {\omega_5}^{38} X^9 + {\omega_5}^{158} X^8 + {\omega_5}^{158} X^7 + {\omega_5}^{158} X^6 + {\omega_5}^{38} X^5 + {\omega_5}^{38} X^4+ {\omega_5}^{38} X^3 + {\omega_5}^{158} X^2 + {\omega_5}^{38} X + 4,$\\
\par
$\alpha_{11}={\omega_5}^{143} X^{18} + {\omega_5}^{11} X^{17} + {\omega_5}^{11} X^{16} + {\omega_5}^{143} X^{15} + {\omega_5}^{143} X^{14} + {\omega_5}^{143} X^{13} +{\omega_5}^{143} X^{12} + {\omega_5}^{11} X^{11} + {\omega_5}^{143} X^{10} + {\omega_5}^{11} X^9 + {\omega_5}^{143} X^8 + {\omega_5}^{11} X^7 +{\omega_5}^{11} X^6 + {\omega_5}^{11} X^5 + {\omega_5}^{11} X^4 + {\omega_5}^{143} X^3 + {\omega_5}^{143} X^2 + {\omega_5}^{11} X + 8,$\\
\par
$\alpha_{12}={\omega_6}^{40} X^6 + {\omega_6}^{40} X^5 + {\omega_6}^{104} X^4 + {\omega_6}^{40} X^3 + {\omega_6}^{104} X^2 + {\omega_6}^{104} X+ 15,$\\
\par
$\alpha_{13}={\omega_6}^{19} X^{10} + {\omega_6}^{35} X^9 + {\omega_6}^{19} X^8 + {\omega_6}^{19} X^7 + {\omega_6}^{19} X^6 + {\omega_6}^{35} X^5 + {\omega_6}^{35} X^4 + {\omega_6}^{35} X^3 + {\omega_6}^{19} X^2 + {\omega_6}^{35} X + 2,$\\
\par
$\alpha_{14}={\omega_7}^{61} Y^6 + {\omega_7}^{61} Y^5 + {\omega_7}^{79} Y^4 + {\omega_7}^{61} Y^3 + {\omega_7}^{79} Y^2 + {\omega_7}^{79} Y + 14,$\\
\par
$\alpha_{15}={\omega_7}^{331} X^{10} + {\omega_7}^{169} X^9 + {\omega_7}^{331} X^8 + {\omega_7}^{331} X^7 + {\omega_7}^{331} X^6 + {\omega_7}^{169} X^5 +{\omega_7}^{169} X^4 + {\omega_7}^{169} X^3 + {\omega_7}^{331} X^2 + {\omega_7}^{169} X + 16,$\\
\par
\noindent
$\omega_1$ is a primitive element of $\mathbb{F}_{2^2}={{{\mathbb{F}_2}[X]} \mathord{\left/
 {\vphantom {{{\mathbb{F}_2}[X]} {\left\langle {{X^2} + X + 1} \right\rangle }}} \right.
 \kern-\nulldelimiterspace} {\left\langle {{X^2} + X + 1} \right\rangle }}$ with ${\rm{ord}}(\omega_1 ) = {2^2} - 1 = 3$,\\
  $\omega$ is a primitive element of $\mathbb{F}_{3^2}={{{\mathbb{F}_3}[X]} \mathord{\left/
 {\vphantom {{{F_3}[X]} {\left\langle {X^2+2X+2} \right\rangle }}} \right.
 \kern-\nulldelimiterspace} {\left\langle X^2+2X+2 \right\rangle }}$ with ${\rm{ord}}(\omega ) = {3^2} - 1 = 8$,\\
 $\omega_2$ is a primitive element of $\mathbb{F}_{5^2}={{{\mathbb{F}_5}[X]} \mathord{\left/
 {\vphantom {{{F_5}[X]} {\left\langle {X^2+4X+2} \right\rangle }}} \right.
 \kern-\nulldelimiterspace} {\left\langle X^2+4X+2 \right\rangle }}$ with ${\rm{ord}}(\omega_2 ) = {5^2} - 1 = 24$,\\
 $\omega_3$ is a primitive element of $\mathbb{F}_{7^2}={{{\mathbb{F}_7}[X]} \mathord{\left/
 {\vphantom {{{F_7}[X]} {\left\langle {X^2+6X+3} \right\rangle }}} \right.
 \kern-\nulldelimiterspace} {\left\langle X^2+6X+3 \right\rangle }}$ with ${\rm{ord}}(\omega_3 ) = {7^2} - 1 = 48$,\\
 $\omega_4$ is a primitive element of $\mathbb{F}_{11^2}={{{\mathbb{F}_{11}}[X]} \mathord{\left/
 {\vphantom {{{F_{11}}[X]} {\left\langle {X^2 + 7X + 2} \right\rangle }}} \right.
 \kern-\nulldelimiterspace} {\left\langle X^2 + 7X + 2 \right\rangle }}$ with ${\rm{ord}}(\omega_4 ) = {11^2} - 1 = 120$,\\
 $\omega_5$ is a primitive element of $\mathbb{F}_{13^2}={{{\mathbb{F}_{13}}[X]} \mathord{\left/
 {\vphantom {{{F_{13}}[X]} {\left\langle {X^2 + 12X + 2} \right\rangle }}} \right.
 \kern-\nulldelimiterspace} {\left\langle X^2 + 12X + 2 \right\rangle }}$ with ${\rm{ord}}(\omega_5 ) = {13^2} - 1 = 168$,\\
 $\omega_6$ is a primitive element of $\mathbb{F}_{17^2}={{{\mathbb{F}_{17}}[X]} \mathord{\left/
 {\vphantom {{{F_{17}}[X]} {\left\langle {X^2 + 16X + 3} \right\rangle }}} \right.
 \kern-\nulldelimiterspace} {\left\langle X^2 + 16X + 3 \right\rangle }}$ with ${\rm{ord}}(\omega_6 ) = {17^2} - 1 = 288$\\
  and
  $\omega_7$ is a primitive element of $\mathbb{F}_{19^2}={{{\mathbb{F}_{19}}[X]} \mathord{\left/
 {\vphantom {{{F_{19}}[X]} {\left\langle {X^2 + 18X + 2} \right\rangle }}} \right.
 \kern-\nulldelimiterspace} {\left\langle X^2 + 18X + 2 \right\rangle }}$ with ${\rm{ord}}(\omega_7 ) = {19^2} - 1 = 360$.

%\begin{acknowledgements}
%\end{acknowledgements}



\end{document}